\documentclass[conference]{IEEEtran}
\IEEEoverridecommandlockouts
\usepackage{amssymb}
\usepackage{hyperref}
\usepackage{cite}
\usepackage{graphicx}
\usepackage{array}
\usepackage{multirow}
\usepackage{amsmath}
\usepackage{stfloats}
\usepackage{graphicx}
\usepackage{subfigure}
\usepackage{tabularx}
\usepackage{booktabs}
\usepackage{epsfig,epsf}
\usepackage{setspace}
\usepackage{bm}
\usepackage{textcomp}
\usepackage{algorithmic}
\usepackage{algorithm}
\usepackage{amsmath}
\usepackage{subfigure}
\usepackage{caption}
\usepackage{graphicx}
\usepackage{setspace}
\usepackage{url}
\usepackage{verbatim}
\usepackage{graphicx}
\usepackage{cite}
\usepackage{amssymb}
\usepackage{color}
\usepackage{xpatch}
\definecolor{shadecolor}{RGB}{0,0,255}
\definecolor{blue}{RGB}{0,0,255}

\DeclareMathOperator{\sinc}{sinc}

\newtheorem{lemma}{Lemma}

\newtheorem{remark}{Remark}
\usepackage{soul}
\usepackage{booktabs} 
\usepackage{tabularx}
\usepackage{array} 
\usepackage[table]{xcolor} 
\definecolor{tableblue}{RGB}{221,238,248} \newcolumntype{Y}{>{\raggedright\arraybackslash}X}

\makeatletter

\ExplSyntaxOn
\cs_new:Npn \bibColoredItems #1#2
{
	\clist_map_inline:nn {#2} { \cs_new:cpn {bib@colored@##1} {#1} } 
}
\ExplSyntaxOff

\newcommand\bib@setcolor[1]{%
	\ifcsname bib@colored@#1\endcsname
	\expanded{\noexpand\color{\csname bib@colored@#1\endcsname}}%
	\else
	\normalcolor
	\fi
}

\makeatother

\begin{document}
	
\title{An Optical Pathway to Movable Rydberg Atomic Quantum Receivers}

\author{Qihao Peng,~\IEEEmembership{Member,~IEEE},
         Qu Luo,~\IEEEmembership{Member,~IEEE},
    Lixia Xiao,~\IEEEmembership{Senior Member,~IEEE},\\
    De Mi,~\IEEEmembership{Senior Member,~IEEE},
    Neng Ye,~\IEEEmembership{Senior Member,~IEEE},
        Cunhua Pan,~\IEEEmembership{Senior Member,~IEEE},\\
      Pei Xiao,~\IEEEmembership{Senior Member,~IEEE}, Cheng-Xiang Wang,~\IEEEmembership{Fellow,~IEEE}, Jiangzhou Wang,~\IEEEmembership{Fellow,~IEEE}.
         
    %
   \vspace{-2.5cm}
		\thanks{Q. Peng, Q. Luo, and P. Xiao are affiliated with 5G and 6G Innovation Centre, Institute for Communication Systems (ICS) of the University of Surrey, Guildford, GU2 7XH, UK. (e-mail: \{q.peng, q.u.luo,p.xiao\}@surrey.ac.uk). } \\
        \thanks{L. Xiao is with  the Research Center of 6G Mobile Communications, School of Cyber Science and Engineering, Huazhong University of Science and Technology, Wuhan 430074, China. (e-mail: lixiaxiao@hust.edu.cn).} \\
         \thanks{D. Mi with the College of Computing, Birmingham City University, Birmingham B4 7XG, U.K. (email: de.mi@bcu.ac.uk).}\\
        \thanks{N. Ye is the School of Cyberspace Science and Technology, Beijing Institute of Technology, Beijing 100081, China. (email: ianye@bit.edu.cn).}\\
        
           \thanks{C. Pan, C.-X. Wang, and J. Wang are with the National Mobile Communications Research Laboratory, Southeast University, Nanjing, China. (e-mail: \{cpan,chxwang,j.z.wang\}@seu.edu.cn).}  \\

        }

	
	\maketitle

\begin{abstract}
This paper develops an optically movable Rydberg atomic quantum receiver (RAQR), in which the probe and coupling beams are steered within each vapor cell to dynamically reconfigure the effective radio-frequency (RF) sensing position without mechanical actuation. A closed-form equivalent baseband model is derived by separating the atomic transduction coefficient, optical steering phase, and cell-center array response into distinct factors and the accuracy of the resulting model is validated against numerical solutions of the Lindblad master equation. Based on the derived model, we reveal two complementary channel-shaping mechanisms, including intrinsic beam-pattern shaping through RF-to-optical transduction and per-cell phase control enabled by optical displacement. To further exploit these capabilities, a non-convex sum-rate maximization problem is formulated over the optical positions and local oscillator design and solved via an alternating optimization framework with analytical gradients. Simulation results validate the derived model and demonstrate substantial performance gains enabled by optical movability, highlighting its potential as a programmable receiver architecture for future wireless networks.
\end{abstract}

\begin{IEEEkeywords}
Rydberg atomic quantum receiver, movable receiver, channel reconfiguration, optical beam steering.
\end{IEEEkeywords}

\section{Introduction}
\label{sec:introduction}

\IEEEPARstart{S}{ixth}-generation (6G) wireless networks are envisioned to support massive connectivity, high-precision sensing, and reliable communications \cite{wang2025modeling,zhu2026movable}. These demands impose unprecedented requirements on receiver front ends. However, conventional radio-frequency (RF) receivers remain fundamentally constrained by the thermal-noise-limited sensitivity of the low-noise amplifier chain and the Chu limit of electrically small antennas \cite{mclean1996re}. Meanwhile, the fixed geometry of conventional antenna arrays limits their adaptability to dynamic wireless environments, hindering real-time channel programmability and spatial reconfiguration.

To enhance the programmability of RF antennas, movable antennas (MAs) and fluid antenna systems (FASs) have been extensively investigated by exploiting the spatial variations of wireless channels, also known as spatial degrees of freedom (DoFs) \cite{FAScomst,Zhu2026COMST}. For channel modeling, the field-response model in \cite{zhu2023modeling} represents the channel as a superposition of path responses with position-dependent phases, while \cite{WongFAS} characterizes FASs through the spatial correlation
among discrete ports within a confined aperture. Since position-dependent channel information is required, dedicated channel estimation methods have been developed for MA/FAS architectures \cite{ma2023compressed,xu2023channel}. The achievable multiple-input-multiple-output (MIMO) capacity with continuous position optimization was studied in \cite{ma2023mimo}, and the multiuser sum-rate gains over
fixed-position arrays were quantified in \cite{Zhu2024TWC}. To solve the resulting non-convex joint position and beamforming problems, gradient descent and alternating optimization methods have been widely adopted \cite{Zhu2024TWC,Tang2025TWC}, followed by graph-based \cite{mei2024movable,11343789} and learning-based \cite{weng2024learning,wang2024ai} approaches. Beyond communication, MAs and FASs have been extended to emerging applications, including physical-layer security \cite{tang2024secure,ghadi2024physical},
integrated sensing and communication \cite{lyu2025movable,wang2024fluid},
and near-field/six-dimensional wireless architectures \cite{Zhunearfield,shao20256d}. Nevertheless, MAs and FASs improve only the spatial flexibility of RF receivers while retaining the fundamental limitations of electronic antenna front ends, such as thermal-noise-limited sensitivity and the Chu limit, which motivates the exploration of receiver paradigms beyond conventional RF architectures.

In response, Rydberg atomic quantum receivers (RAQRs) provide a
fundamentally different reception paradigm by directly mapping incident electromagnetic fields onto atomic quantum states
\cite{Peng2026WCM}. By exciting alkali atoms in a vapor cell to Rydberg states, the incident RF field modulates the optical probe signal through electromagnetically induced transparency (EIT) and Autler-Townes splitting (ATS), which can be directly read out
by a photodetector \cite{sedlacek2012microwave}. Owing to the large dipole moments of Rydberg states, RAQRs enable highly sensitive RF detection with optical readout approaching the quantum noise limit \cite{tu2024approaching}. Unlike conventional RF receivers, the noise floor of an RAQR is no longer determined solely by amplifier-chain thermal noise, but by a combination of thermal noise, photon shot noise, and quantum projection noise \cite{Gongtcom}, while the effective receiving aperture is decoupled from the
physical size of the vapor cell. Consequently, RAQRs overcome the
fundamental sensitivity and aperture constraints of conventional RF
receivers \cite{cox2018quantum}. 

Building on these properties, RAQRs have evolved from single-point
electromagnetic field sensing toward communication-oriented receiver
architectures. For RF-to-optical conversion modeling, a nonlinear
magnitude-only model was first developed in \cite{CuiJASC}, which differs fundamentally from conventional complex-valued MIMO models. To facilitate link-level analysis, equivalent baseband models of RAQRs were further established in \cite{Gongtcom,ChenTWC}, followed by a general dynamic model in \cite{zhu2025general}. Beyond signal modeling, the beamforming capability of a single vapor cell was investigated in \cite{Cuibeam}, while Peng \emph{et al.} further extended this capability through Rydberg arrays with near-field local oscillators (LOs) \cite{peng2026celllevelchannelshapingrydberg}.
Leveraging these unique advantages, RAQRs have been explored in various applications, including direction-of-arrival (DoA) estimation \cite{ChenDOA}, integrated sensing and communication (ISAC) \cite{ZhangMWC}, and simultaneous wireless information and power transfer (SWIPT) \cite{PengJSACSWIPT}.

Despite these advances, existing RAQR architectures do not provide spatial reconfigurability comparable to that of MAs and FASs. Both single-cell beamforming \cite{Cuibeam} and cell-array channel shaping \cite{peng2026celllevelchannelshapingrydberg} operate with fixed probe-coupling overlap regions, thereby limiting the exploitable intra-cell spatial DoFs. Prior vapor-cell experiments have demonstrated that the atomic sensing region can be optically repositioned by scanning the pump--probe intersection within a stationary cell using acousto-optic deflectors \cite{maddox2023rapid}. This experimental evidence establishes the physical basis for optically reconfiguring the effective sensing position and motivates an optically movable RAQR that exploits the intra-cell spatial DoFs through controllable probe--coupling overlap positioning.

In response, this paper proposes an optically movable RAQR architecture and investigates its capability for channel reshaping. Furthermore, to quantify the achievable decorrelation gain, we maximize the maximum ratio combining (MRC)-based sum rate through a dedicated optimization algorithm. The main contributions are summarized as follows:

\begin{itemize}
\item 
An optically movable RAQR architecture is proposed. Unlike conventional MAs and FASs, which rely on mechanical displacement or port switching for spatial reconfiguration, the proposed architecture reconfigures the optical position by steering the probe and coupling beams, thereby enabling continuous spatial reconfiguration without mechanical actuation.

\item 
We derive a closed-form equivalent baseband model for the movable RAQR, where the atomic transduction coefficient, optical steering phase, and array response are explicitly separated into distinct factors, rendering the effective channel an explicit and differentiable function of the optical positions. The accuracy of the adopted center approximation is further validated against the numerical solutions based on the Lindblad master equation.

\item We reveal two complementary channel-shaping mechanisms of the optically movable RAQR. The RF-to-optical transduction provides intrinsic beam-pattern shaping through phase matching, while optical displacement enables per-cell phase reconfiguration for suppressing inter-user interference, beyond the capability of the real-valued transduction coefficient. The achievable phase range and the conditions for two-user channel decorrelation are characterized.

\item We formulate a joint optical-position and LO-design optimization problem for non-convex MRC-based sum-rate maximization and develop an alternating optimization framework based on analytical gradients. Simulation results demonstrate significant user-decorrelation gains and the superiority of optically movable RAQRs over fixed-position architectures.
\end{itemize}

The remainder of this paper is organized as follows.
Section II presents the system model and the superimposed RF signal. Section III derives the closed-form expressions for the optically movable RAQR. Section IV develops the
closed-form gradients and the alternating optimization algorithm. The simulation results are presented in 
Section V, and conclusions are drawn in
Section VI.

\textit{Notations:} The operators $(\cdot)^{\mathrm{T}}$, $(\cdot)^{\mathrm{H}}$, $(\cdot)^{*}$, and $(\cdot)^{-1}$ denote the transpose, Hermitian transpose, complex conjugate, and inverse, respectively. The determinant, trace, and expectation operators are denoted by $\det[\cdot]$, $\mathrm{tr}(\cdot)$, and $\mathbb{E}\{\cdot\}$, respectively. The notation $\mathrm{diag}\{\cdot\}$ forms a diagonal matrix, $\odot$ denotes the Hadamard product, and $[\mathbf{A}]_{r,k}$ denotes the $(r,k)$-th entry of $\mathbf{A}$. $\Re\{\cdot\}$, $\Im\{\cdot\}$, $|\cdot|$, and $\arg\{\cdot\}$ denote the real part, imaginary part, magnitude, and phase, respectively. The Euclidean norm is denoted by $\|\cdot\|$. $\mathbf{I}_N$ denotes the $N\times N$ identity matrix. The notation $\mathcal{CN}(\boldsymbol{\mu},\mathbf{R})$ denotes a circularly symmetric complex Gaussian distribution with mean $\boldsymbol{\mu}$ and covariance matrix $\mathbf{R}$. Furthermore, $(\cdot)'$ denotes differentiation, $j=\sqrt{-1}$ is the imaginary unit, and $c$, $\epsilon_0$, $q$, and $a_0$ denote the speed of light, vacuum permittivity, elementary charge, and Bohr radius, respectively.

\section{System Model}
In this section, we describe the models of the Rydberg array and near-field LO. Then, we derive the output probe beam based on the superimposed RF signal. 

\subsection{Model of Rydberg Array}
\begin{figure}
    \centering
    \includegraphics[width=1\linewidth]{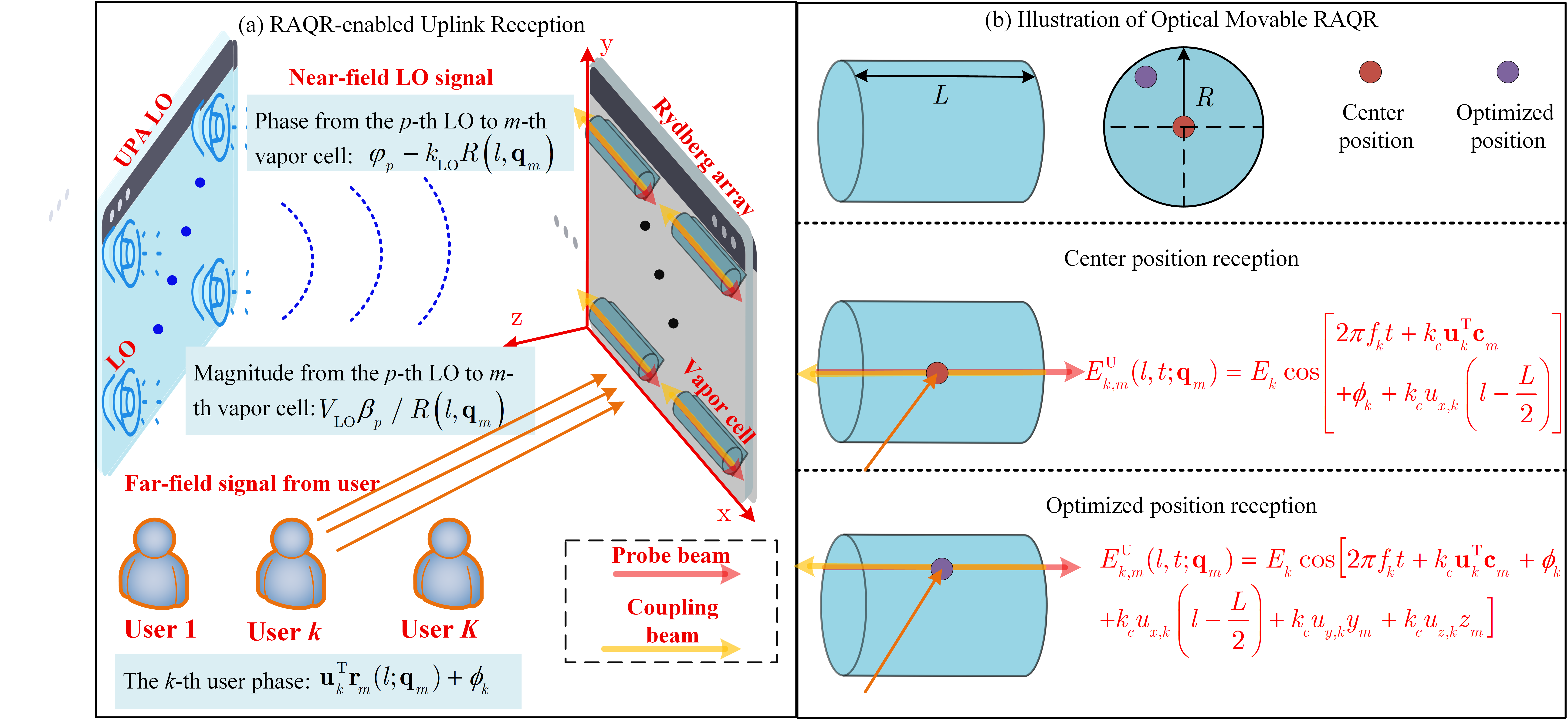}
    \caption{Illustration of optically movable RAQR.}
    \label{mRAQR}
     \vspace{-0.5cm}
\end{figure}
As shown in Fig. \ref{mRAQR}, we consider an \(M = M_x\times M_y\) Rydberg vapor cell array, where each vapor cell is modeled as a cylinder with length \(L\) and radius \(R_{\rm VC}\). The longitudinal axis of each vapor cell is parallel to the \(x\)-axis. The edge-to-edge spacing between adjacent vapor cells along the \(x\)-axis is \(d_x\), and the center-to-center spacing along the \(y\)-axis is \(d_y\). For ease of derivation, the two-dimensional index \((m_x,m_y)\) can be mapped into a one-dimensional index \(m=(m_y-1)M_x+m_x\). For the vapor cell indexed by \((m_x,m_y)\), its axial position is modeled as 
\begin{equation} \small
\mathbf r_m(l) = 
\begin{bmatrix} (m_x-1)(L+d_x)+l\\ (m_y-1)d_y\\ 0 
\end{bmatrix}, \qquad 0\le l\le L .
\end{equation} 
Accordingly, the center of the \((m_x,m_y)\)-th vapor cell is
\begin{equation} \small
\mathbf c_m = 
\begin{bmatrix} (m_x-1)(L+d_x)+\frac{L}{2}\\ (m_y-1)d_y\\ 0 
\end{bmatrix}. 
\end{equation}
To realize optically movable Rydberg sensing, the probe and coupling beams counter-propagate along the \(x\)-axis inside each vapor cell, and the EIT readout occurs only within their spatial overlap. Therefore, by judiciously designing the optical beams, the sensing region can be flexibly repositioned within each vapor cell, thereby enabling a programmable receive-side sensing response\footnote{A feasible implementation is to split the probe and coupling lasers into \(M\) cell-specific optical branches using fiber splitters, with each branch equipped with a two-dimensional micro-electro-mechanical systems (MEMS) steering mirror. For the \(m\)-th cell, the probe and coupling beams are injected from opposite ends and jointly steered to overlap at \(\mathbf q_m=[0,y_m,z_m]^{\mathrm T}\).}. Since the optical beam radius is much smaller than both the RF wavelength and the radius of the vapor cell, the probe-coupling overlap can be approximated as a thin sensing line. We denote \(\mathbf q_m=[0,y_m,z_m]^{\rm T}\) as the displacement of the optical position of the \(m\)-th cell, and the overlapping region  becomes
\begin{equation}\label{sensline}\small
\mathbf r_{m}(l;\mathbf{q}_m) = 
\begin{bmatrix} (m_x-1)(L+d_x)+l\\ (m_y-1)d_y+y_{m}\\ z_m
\end{bmatrix},
\end{equation} 
where the optical displacement satisfies \(y^2_m + z^2_m \le R_{\rm VC}^2\).

\subsection{Far-field User Electric Field}
We consider $K$ single-antenna users located in the far field of the Rydberg receiver, and thus the incident field from each user can be approximated as a plane wave over the entire vapor cell array. Let \(\mathbf u_k  = [u_{x,k},u_{y,k},u_{z,k}]^{\rm T}\) denote the directional vector of the \(k\)-th user, where \(u_{x,k}\), \(u_{y,k}\), and \(u_{z,k}\) are the projections of the incident direction onto the \(x\)-, \(y\)-, and \(z\)-axes, respectively. 
Based on the optical position in \eqref{sensline}, the electric field of the \(k\)-th user at the \((m_x,m_y)\)-th vapor cell is modeled as
\begin{equation}\small
\begin{split}
        E_{k,m}^{\rm U}
    (l,t;\mathbf{q}_m)
    &=
    E_k
    \cos
    \Big[
    2\pi f_k t +
    k_c
    \mathbf u_k^{\rm T}
    \mathbf r_{m}(l;\mathbf{q}_m)
    +
    \phi_k
    \Big],
\end{split}
\end{equation}
where \(E_k\), \(f_k\), \(k_c\), and \(\phi_k=\phi_{k,s} + \phi_{k,\rm ch}\) denote the received electric-field amplitude, the Doppler-shifted carrier frequency, the RF wavenumber, and the phase caused by the transmitted symbol \(\phi_{k,s}\) and channel \(\phi_{k,\rm ch}\), respectively. Specifically, \(f_k=f_c+f_{D,k}\), where \(f_c\) is the common carrier frequency and \(f_{D,k}\) is the Doppler shift of the \(k\)-th user. Since \(|f_{D,k}|/f_c\ll1\), a common wavenumber \(k_c=2\pi f_c/c\) is adopted for all users.

Equivalently, using the vapor cell center,
the optical position can be rewritten as
\begin{equation}
    \mathbf r_{m}(l;\mathbf{q}_m)
    =
    \mathbf c_{m}
    +
    \big(l-\frac{L}{2}\big)\mathbf e_x
    +
    y_{m}\mathbf e_y+z_m\mathbf{e}_z,
\end{equation}
where \(\mathbf{e}_x = [1,0,0]^T\), \(\mathbf{e}_y = [0,1,0]^T\), and \(\mathbf{e}_z = [0,0,1]^T\) denote the unit vectors along the x-, y-, and z-axes, respectively. Therefore, the user electric field can also be expressed as
\begin{equation}\small
\begin{split}
     & E_{k,m}^{\rm U}
    (l,t;\mathbf{q}_m)
   =
    E_k
    \cos
    \Big[
    2\pi f_k t
    +
    k_c
    \mathbf u_k^{\rm T}
    \mathbf c_{m}+
    \phi_k
    \\
    &
    +
    k_cu_{x,k}
    \left(l-\frac{L}{2}\right)
    +
    k_cu_{y,k}y_{m} +   k_cu_{z,k}z_{m} 
    \Big].
\end{split}
\end{equation}
This expression shows that optical repositioning changes the spatial phase of the far-field user signal through the term \(k_cu_{y,k}y_{m}\) and \(k_cu_{z,k}z_{m}\), while the received field amplitude \(E_k\) remains approximately unchanged due to the long-distance propagation. Therefore, the movable probe-coupling overlap mainly provides phase-domain channel reconfiguration for far-field users.

\subsection{Near-field LO Electric Field}
The LO is generated by a \(P=P_xP_y\)-element uniform planar array (UPA) placed in the near field of the Rydberg vapor cell array. The LO UPA is parallel to the \(x\)-\(y\) plane.  For ease of notation, we map the two-dimensional LO index \((p_x,p_y)\) into a one-dimensional index \(p=(p_y-1)P_x+p_x\), \(p=1,\ldots,P\). The \((p_x,p_y)\)-th LO element is located at
\begin{equation}\small
\mathbf r_{{\rm LO},p} = \begin{bmatrix} x_{\rm LO}+(p_x-1)d_{{\rm LO},x}\\ y_{\rm LO}+(p_y-1)d_{{\rm LO},y}\\ z_{\rm LO} 
\end{bmatrix},
\label{eq:lo_upa_position} 
\end{equation} where \(x_{\rm LO}\), \(y_{\rm LO}\), and \(z_{\rm LO}\) denote the reference coordinates of the LO UPA. The inter-element spacings are set as \(d_{{\rm LO},x}=d_{{\rm LO},y}=\frac{\lambda_{\rm LO}}{2}\), where \(\lambda_{\rm LO}\) is the LO wavelength.
The coefficient of the \(p\)-th LO element is \(a_p=\beta_p e^{{\rm j}\varphi_p}\), where \(\beta_p\) and \(\varphi_p\) denote the controllable amplitude and phase, respectively. The LO vector is denoted by \(\mathbf a=[a_1,\ldots,a_P]^{\rm T}\). Then, the distance between the \(p\)-th LO element and the sensing point in the \((m_x,m_y)\)-th vapor cell is
\begin{equation}\small
\begin{split}      
R_{p,m}(l;\mathbf{q}_{m})    & =    \left\|    \mathbf r_{m}(l;\mathbf{q}_{m})    -    \mathbf r_{{\rm LO},p}    \right\|_2\end{split}
\end{equation} 
Therefore, the LO signal received at the \(m\)-th vapor cell is 
\begin{equation}\small
\begin{split}
       &  \mathcal E_{{\rm LO},m}
    (l,t;\mathbf{q}_{m},\mathbf a)
  =
    V_{\rm LO}\sum_{p=1}^{P}
    \frac{\beta_p}
    {R_{p,m}(l;\mathbf{q}_{m})}\cos\big[2\pi f_{\rm LO}t- k_{\rm LO} \\
    &\times
    R_{p,m}(l;\mathbf{q}_{m})+\varphi_p\big] = \Re\left\{ \mathcal A_{{\rm LO},m}(l;\mathbf{q}_{m},\mathbf a)
    e^{{\rm j}2\pi f_{\rm LO}t}\right\},
\end{split}
    \label{eq:near_field_lo_ula}
\end{equation}
where \(V_{\rm LO} = \sqrt{60P_{\rm LO}G_{\rm LO}}\) is the LO's voltage with respect to transmission power \(P_{\rm LO}\) and transmission gain \(G_{\rm LO}\). \(f_{\rm LO}\) denotes the LO frequency, and \(k_{\rm LO}=2\pi/\lambda_{\rm LO}\) denotes the LO wavenumber. \(  \mathcal A_{{\rm LO},m}(l;\mathbf{q}_{m},\mathbf a)\) is given by

\begin{equation}\small
    \mathcal A_{{\rm LO},m}(l;\mathbf{q}_{m},\mathbf a)
    =
    V_{\rm LO}\sum_{p=1}^{P}
    \frac{\beta_p}{R_{p,m}(l;\mathbf{q}_{m})}\,
    e^{{\rm j}\left[\varphi_p-k_{\rm LO}R_{p,m}(l;\mathbf{q}_{m})\right]}.
    \label{eq:lo_phasor}
\end{equation}
The LO's magnitude and phase are defined as
\begin{equation}\small
    E_{{\rm LO},m}(l;\mathbf{q}_{m},\mathbf a)=\left|\mathcal A_{{\rm LO},m}(l;\mathbf{q}_{m},\mathbf a)\right|,
    \label{eq:lo_amp_phase}
\end{equation}
and
\begin{equation}\small
    \Phi_{{\rm LO},m}(l;\mathbf{q}_{m},\mathbf a)=\arg\left\{\mathcal A_{{\rm LO},m}(l;\mathbf{q}_{m},\mathbf a)\right\}.
\end{equation}

\subsection{Superimposed RF Electric Field}
Using the LO as the reference, the superimposed RF field of the \((m_x,m_y)\)-th vapor cell with the optical position \(\mathbf{q}_m\) is
\begin{equation}\small
\begin{split}
     & \mathcal E_{{\rm syn},m}
    (l,t;\mathbf{q}_{m},\mathbf a)
    =
     \mathcal E_{{\rm LO},m}
    (l,t;\mathbf{q}_{m},\mathbf a)
    +
    \sum_{k=1}^{K}
    E_k\cos\big[2\pi f_{k}t\\
    &+
    k_c\mathbf u_k^{\rm T}
    \mathbf r_{m}(l;\mathbf{q}_{m})
    +
    \phi_k\big] = \Re\{E_{{\rm RF},m}(l,t;\mathbf{q}_{m},\mathbf a)e^{{\rm j}2\pi f_{\rm LO}t}\}.
\end{split}
    \label{eq:syn_complex_field}
\end{equation}
Under the strong-LO condition, i.e., \( E_{{\rm LO},m}(l;\mathbf{q}_{m},\mathbf a)\gg E_k \), \(\forall k\), we have 
\begin{equation}\small
\begin{split}
  E_{{\rm RF},m}(l,t;\mathbf{q}_{m},\mathbf a)
    & \approx
    E_{{\rm LO},m}
    (l;\mathbf{q}_{m},\mathbf a)
    +
    \sum_{k=1}^{K}
    E_k
    \cos
    \big[
    2\pi \Delta f_{k}t\\
    &+
    \Delta\Phi_{m,k}
    (l;\mathbf{q}_{m},\mathbf a)
    \big],
\end{split}
    \label{eq:syn_field_strength_strong_lo}
\end{equation}
where \( \Delta f_{k} = f_k - f_{\rm LO}\). Here, we assume \(f_{\rm LO} = f_c\), and thus we have \(\Delta f_{k}=f_{D,k}\). The relative phase between the \(k\)-th user field and the LO field is
\begin{equation}\small
\begin{aligned}
    \Delta\Phi_{m,k}
    (l;\mathbf{q}_{m},\mathbf a)
    =
    &
    k_c\mathbf u_k^{\rm T}
    \mathbf r_{m}(l;\mathbf{q}_{m})
   \! +\!
    \phi_k\!-\!
    \Phi_{{\rm LO},m}
    (l;\mathbf{q}_{m},\mathbf a).
\end{aligned}
    \label{eq:user_lo_relative_phase}
\end{equation}

\subsection{Four-Level Ladder System}

We consider a four-level ladder system
$|1\rangle\rightarrow|2\rangle\rightarrow|3\rangle\rightarrow|4\rangle$, where the Rabi frequencies are defined as \(\Omega_{p} \triangleq  \Omega_{p,m}\), \(\Omega_{c} \triangleq \Omega_{c,m}\), \(\Omega_{{\rm RF},m}(l,t)=\mu_{34}E_{{\rm RF},m}(l,t;\mathbf{q}_{m},\mathbf a)/\hbar\). \(\mu_{ij}\) is the dipole moment between \(|i\rangle\rightarrow|j\rangle\). The dynamics of the four-level transition scheme are governed by the Lindblad master equation, which is given by 
\begin{equation}\small
      \frac{{\rm d}\boldsymbol{\rho} }{{\rm d} t} = -{\rm j} [\mathbf{H},\boldsymbol{\rho}]
    - \frac{1}{2}\{\boldsymbol{\Gamma},\boldsymbol{\rho} \} + \boldsymbol{\Lambda},
\end{equation}
where \(\mathbf{H}\), \(\boldsymbol{\Gamma}\), and \(\boldsymbol{\Lambda}\)
represent the Hamiltonian, the relaxation matrix, and the repopulation matrix, respectively \footnote{The detailed definitions can be found in \cite{Gongtcom}.}.
By setting \( \frac{\text{d}\boldsymbol{\rho} }{\text{d}t}  = 0\), we obtain the steady-state solution of density matrix \(\boldsymbol{\rho}\). Then, the measured susceptibility based on the probe beam is proportional to $\rho_{21}$, which is given by \cite{Gongtcom}
\begin{equation}\small\label{rho21}
    \begin{split}
 & \rho_{21}(\Omega_{\mathrm{RF}})=\Omega_{p}\times \\
 & \frac{A_{1}\Omega_{\mathrm{RF}}^{4}+A_{2}\Omega_{\mathrm{RF}}^{2}+A_{3}-{\rm j}\left(B_{1}\Omega_{\mathrm{RF}}^{4}+B_{2}\Omega_{\mathrm{RF}}^{2}+B_{3}\right)}{C_{1}\Omega_{\mathrm{RF}}^{4}+C_{2}\Omega_{\mathrm{RF}}^{2}+C_{3}},
    \end{split}
\end{equation}
where \(A_1\), \(A_2\), \(A_3\), \(B_1\), \(B_2\), \(B_3\), \(C_1\), \(C_2\), \(C_3\) can be found in Appendix A of \cite{Gongtcom}. With the given RF signal, the output probe power of the \(m\)-th cell is
given by
\begin{equation}\small\label{outprobe}
\begin{split}
        P_{{\rm out},m}(t) &= P_{{\rm in},m}  \exp\Big\{k_p D_{\Omega}\int^L_0
        \Im\Big\{\rho_{21}\big[\Omega_{{\rm RF},m}(l,t)\big]\Big\}{\rm d}l  \Big\},
\end{split}
\end{equation}
where \(P_{{\rm in},m}\) is the input power of the probe beam,
\(D_{\Omega} = -\frac{2 N_0 \mu^2_{12}}{\epsilon_0\hbar \Omega_p}\) is a
constant with \(N_0\) denoting the atomic number density and \(\epsilon_0\)
the vacuum permittivity, and \(k_p = \frac{2\pi f_p}{c}\) is the probe
wavenumber with \(f_p\) denoting the probe optical frequency. 

\section{Equivalent Channel Model}
In this section, we derive the closed-form expression of the output probe power and the equivalent baseband channel model of the movable
vapor cell array RAQR.
\subsection{Vapor Cell-Center Approximation}
It is challenging to derive the closed-form expression of the output probe power, as substituting this position-dependent field directly into the atomic response yields an intractable integral. To address this issue, we adopt a cell-center approximation. Particularly, the field amplitude is treated as a constant evaluated at the center of each vapor cell, while the phase variation is retained via a first-order Taylor expansion.

The center of the vapor cell is 
\begin{equation}\small
    \mathbf r_m^{c}(\mathbf{q}_{m})
    =
    \mathbf r_m\Big(\frac{L}{2};\mathbf{q}_{m}\Big)
    =
    \mathbf c_m
    +
    y_m\mathbf e_y+z_m\mathbf{e}_z.
    \label{eq:sensing_line_center}
\end{equation}
Based on the LO phasor \eqref{eq:lo_phasor}, the center values of the LO phasor, field strength, and phase are defined as
\begin{equation}\small
    \mathcal A_{{\rm LO},m}^{c}(\mathbf{q}_{m},\mathbf a)
    = \mathcal A_{{\rm LO},m}\Big(\frac{L}{2};\mathbf{q}_{m},\mathbf a\Big),
    \label{eq:lo_center_envelope}
\end{equation}
\begin{equation}\small
\begin{split}
        E_{{\rm LO},m}^{c}(\mathbf{q}_{m},\mathbf a)
    &= \left|\mathcal A_{{\rm LO},m}^{c}(\mathbf{q}_{m},\mathbf a)\right|,
\end{split}
    \label{eq:lo_center_am}
\end{equation}
\begin{equation}\small
\begin{split}
    \Phi_{{\rm LO},m}^{c}(\mathbf{q}_{m},\mathbf a)
    &= \arg\left\{\mathcal A_{{\rm LO},m}^{c}(\mathbf{q}_{m},\mathbf a)\right\}.
\end{split}
    \label{eq:lo_center_phase}
\end{equation}

Since the length of the vapor cell is smaller than the propagation distance between LO and the vapor cell, the LO field strength along the vapor cell is approximated as
\begin{equation}\small
    E_{{\rm LO},m}(l;\mathbf{q}_{m},\mathbf a)
    \approx
    E_{{\rm LO},m}^{c}(\mathbf{q}_{m},\mathbf a),
    \qquad
    0\le l\le L.
    \label{eq:lo_strength_center_approx}
\end{equation}
However, the phase variation along the vapor cell's length cannot be neglected, since the cell length is comparable to the RF wavelength. By applying the first-order Taylor expansion around \(l=L/2\), the LO phase is approximated as
\begin{equation}\small
    \Phi_{{\rm LO},m}(l;\mathbf{q}_{m},\mathbf a)
    \approx
    \Phi_{{\rm LO},m}^{c}(\mathbf{q}_{m},\mathbf a)
    +
    \eta_{{\rm LO},m}^{c}(\mathbf{q}_{m},\mathbf a)
    \big(l-\frac{L}{2}\big),
    \label{eq:lo_phase_taylor}
\end{equation}
where \(\eta_{{\rm LO},m}^{c}(\mathbf{q}_{m},\mathbf a)\) is written as \eqref{eq:lo_phase_slope_explicit_compact}, given at the bottom of the next page. In \eqref{eq:lo_phase_slope_explicit_compact}, \(\Delta x_{p,m} = (m_x-1)(L+d_x)+\frac{L}{2} - x_{\rm LO}-(p_x-1)d_{{\rm LO},x}\).
\begin{figure*}[bh]
    \hrulefill
    \begin{equation}\small
\begin{aligned}\small
\eta_{{\rm LO},m}^{c}(\mathbf q_{m},\mathbf a)
= \Im \Bigg\{ \frac{ \displaystyle\sum_{p=1}^{P} \beta_p\,
e^{{\rm j} \left[ \varphi_p - k_{\rm LO}R_{p,m}^{c}(\mathbf q_m) \right]}
\left[ - \frac{\Delta x_{p,m}} {\left[R_{p,m}^{c}(\mathbf q_m)\right]^3}
- {\rm j}\,k_{\rm LO} \frac{\Delta x_{p,m}}
{\left[R_{p,m}^{c}(\mathbf q_m)\right]^2} \right] }
{ \displaystyle\sum_{p=1}^{P} \frac{\beta_p}
{R_{p,m}^{c}(\mathbf q_m)}\,
e^{{\rm j} \left[ \varphi_p - k_{\rm LO}R_{p,m}^{c}(\mathbf q_m) \right]} }
\Bigg\}
.
\end{aligned}
\label{eq:lo_phase_slope_explicit_compact}
\end{equation}
\end{figure*}
Similarly, the spatial phase of the \(k\)-th far-field user at the sensing-line center is
\begin{equation}\small
    \Theta_{m,k}^{c}(\mathbf{q}_{m})
    =
    k_c\mathbf u_k^{\rm T}\mathbf r_m^{c}(\mathbf{q}_{m})
    +
    \phi_k .
    \label{eq:user_center_phase}
\end{equation}
Since the far-field user phase is linear with respect to the sensing position, its phase along the sensing line can be exactly expressed as
\begin{equation}\small
    \Theta_{m,k}(l;\mathbf{q}_{m})
    =
    \Theta_{m,k}^{c}(\mathbf{q}_{m})
    +
    k_cu_{x,k}
    \big(l-\frac{L}{2}\big).
    \label{eq:user_phase_line}
\end{equation}

Therefore, the relative phase between the \(k\)-th user signal and the LO field along the sensing line is approximated as
\begin{equation}\small
\begin{aligned}\small
    \Delta\Phi_{m,k}(l;\mathbf{q}_{m},\mathbf a)
    &=
    \Theta_{m,k}(l;\mathbf{q}_{m})
    -
    \Phi_{{\rm LO},m}(l;\mathbf{q}_{m},\mathbf a)
    \\
    &\approx
    \Delta\Phi_{m,k}^{c}(\mathbf{q}_{m},\mathbf a)
    +
    \kappa_{m,k}^{c}(\mathbf{q}_{m},\mathbf a)
    \big(l-\frac{L}{2}\big),
\end{aligned}
    \label{eq:relative_phase_taylor}
\end{equation}
where the center relative phase is
\begin{equation}\small
    \Delta\Phi_{m,k}^{c}(\mathbf{q}_{m},\mathbf a)
    =
    \Theta_{m,k}^{c}(\mathbf{q}_{m})
    -
    \Phi_{{\rm LO},m}^{c}(\mathbf{q}_{m},\mathbf a),
    \label{eq:center_relative_phase}
\end{equation}
and the effective longitudinal phase slope is
\begin{equation}\small
    \kappa_{m,k}^{c}(\mathbf{q}_{m},\mathbf a)
    =
    k_cu_{x,k}
    -
    \eta_{{\rm LO},m}^{c}(\mathbf{q}_{m},\mathbf a).
    \label{eq:effective_phase_slope}
\end{equation}

Consequently, under the strong-LO condition, the RF envelope strength along the sensing line is approximated as
\begin{equation}\small
\begin{aligned}\small
   E_{{\rm RF},m}(l,t;\mathbf{q}_{m},\mathbf a)
    \approx & 
    E_{{\rm LO},m}^{c}(\mathbf{q}_{m},\mathbf a)
    +
    \sum_{k=1}^{K}
    E_k
    \cos
    \Big[
    2\pi f_{D,k}t\\
    + &
    \Delta\Phi_{m,k}^{c}(\mathbf{q}_{m},\mathbf a)
    +\kappa_{m,k}^{c}(\mathbf{q}_{m},\mathbf a)
    (l-\frac{L}{2})
    \Big].
\end{aligned}
    \label{eq:rf_envelope_center_approx}
\end{equation}

\subsection{Closed-Form Probe Output}
Substituting \eqref{eq:rf_envelope_center_approx} into the RF Rabi frequency \(\Omega_{{\rm RF},m}(l,t)=\mu_{34}E_{{\rm RF},m}(l,t;\mathbf{q}_{m},\mathbf a)/\hbar\) defined in Section~II-E, 
we define the LO-biased operating point as \(\Omega_{{\rm LO},m}^{c}(\mathbf{q}_{m},\mathbf a)
    =
    \frac{\mu_{34}}{\hbar}
    E_{{\rm LO},m}^{c}(\mathbf{q}_{m},\mathbf a)\)
and define the user-induced RF Rabi perturbation as \(\Delta\Omega_m(l,t;\mathbf{q}_{m},\mathbf a)\), which is expressed as
\begin{equation}\small
\begin{aligned}\small
    \Delta\Omega_m(l,t;\mathbf{q}_{m},\mathbf a)
    &=
    \sum_{k=1}^{K}
    \Omega_k
    \cos
    \Big[
    2\pi f_{D,k}t
    +
    \Delta\Phi_{m,k}^{c}(\mathbf{q}_{m},\mathbf a)
    \\
    &+
    \kappa_{m,k}^{c}(\mathbf{q}_{m},\mathbf a)
    \big(l-\frac{L}{2}\big)
    \Big],
\end{aligned}
    \label{eq:delta_rabi_closed}
\end{equation}
where \( \Omega_k
    =
    \frac{\mu_{34}E_k}{\hbar}\), \(\forall k\). Owing to \(\Omega_k\ll \Omega_{{\rm LO},m}^c(\mathbf{q}_m,\mathbf{a})\), we use the first-order Taylor expansion and have
\begin{equation}\small
\begin{aligned}\small
    \Im
    \left\{
    \rho_{21}
    \left[
    \Omega_{{\rm RF},m}(l,t)
    \right]
    \right\}
    &\approx
    \Im
    \left\{
    \rho_{21}
    \left[
    \Omega_{{\rm LO},m}^{c}(\mathbf{q}_{m},\mathbf a)
    \right]
    \right\}
    \\
    &
    +
    g_m^{c}(\mathbf{q}_{m},\mathbf a)
    \Delta\Omega_m(l,t;\mathbf{q}_{m},\mathbf a),
\end{aligned}
    \label{eq:rho_im_linearization}
\end{equation}
where \(g_m^{c}(\mathbf{q}_{m},\mathbf a)\) is given by \eqref{eq:gm_closed}.
\begin{figure*}[bh]
    \hrulefill
    \begin{equation}\small
\begin{aligned}\small
   g_m^{c}(\mathbf{q}_{m},\mathbf a)=
    -2\Omega_p\Omega_{{\rm LO},m}^{c}(\mathbf{q}_{m},\mathbf a)
    \frac{
    (B_1C_2-B_2C_1)[\Omega_{{\rm LO},m}^{c}(\mathbf{q}_{m},\mathbf a)]^4
    +
    2(B_1C_3-B_3C_1)[\Omega_{{\rm LO},m}^{c}(\mathbf{q}_{m},\mathbf a)]^2
    +
    (B_2C_3-B_3C_2)
    }
    {
    \left[
    C_1[\Omega_{{\rm LO},m}^{c}(\mathbf{q}_{m},\mathbf a)]^4
    +
    C_2[\Omega_{{\rm LO},m}^{c}(\mathbf{q}_{m},\mathbf a)]^2
    +
    C_3
    \right]^2
    } .
\end{aligned}
    \label{eq:gm_closed}
\end{equation}
\end{figure*}

Substituting \eqref{eq:rho_im_linearization} into \eqref{outprobe}, we have
\begin{equation}\small
\begin{aligned}\small
    P_{{\rm out},m}(t)
      &\approx
    P_{{\rm in},m}
    \exp\Big\{
    k_pD_{\Omega}L
    \Im
    \big\{
    \rho_{21}
    \big[
    \Omega_{{\rm LO},m}^{c}(\mathbf q_{m},\mathbf a)
    \big]
    \big\}
    \Big\}
    \\
    &\times
    \exp\Big\{
    k_pD_{\Omega}\,g_m^{c}(\mathbf q_{m},\mathbf a)
    \int_{0}^{L}
    \Delta\Omega_m(l,t;\mathbf q_{m},\mathbf a)\,{\rm d}l
    \Big\}.
\end{aligned}
    \label{eq:pout_two_exp}
\end{equation}
We define the LO-biased DC output probe power as
\begin{equation}\small
     P_{0,m}^{c}(\mathbf{q}_{m},\mathbf{a})=P_{{\rm in},m}
    \exp{
    \Big\{
    k_pD_{\Omega}L
    \Im
    \left\{
    \rho_{21}
    \left[
    \Omega_{{\rm LO},m}^{c}(\mathbf{q}_{m},\mathbf a)
    \right]
    \right\}
    \Big\}}.
    \label{eq:p0_def_closed}
\end{equation}
The remaining integral in \eqref{eq:pout_two_exp} is calculated as
\begin{equation}\small
\begin{aligned}\small
  \int_{0}^{L}
  \Delta\Omega_m(l,t;\mathbf q_{m},\mathbf a)\,{\rm d}l
    & =
    L\sum_{k=1}^{K}\Omega_k \sinc
    \Big[
    \frac{\kappa_{m,k}^{c}(\mathbf q_{m},\mathbf a)L}{2}
    \Big]\\
    &\times
    \cos
    \left[
    2\pi f_{D,k}t
    +
    \Delta\Phi_{m,k}^{c}(\mathbf q_{m},\mathbf a)
    \right],
\end{aligned}
    \label{eq:cos_integral_closed}
\end{equation}
where \(\sinc(x)=\sin(x)/x\). Therefore, the output probe power can be approximated as
\begin{equation}\small
\begin{aligned}\small
   & P_{{\rm out},m}(t)
    \approx
     P_{0,m}^{c}(\mathbf{q}_{m},\mathbf{a})
    \exp
    \Big\{
    k_pD_{\Omega}g_m^{c}(\mathbf{q}_{m},\mathbf a)L
    \sum_{k=1}^{K}
    \Omega_k
   \\
    &\times   \sinc \Big[
    \frac{\kappa_{m,k}^{c}(\mathbf{q}_{m},\mathbf a)L}{2}
    \Big]
    \cos
    \left[
    2\pi f_{D,k}t
    +
    \Delta\Phi_{m,k}^{c}(\mathbf{q}_{m},\mathbf a)
    \right]
    \Big\}.
\end{aligned}
    \label{eq:pout_closed_exp}
\end{equation}
For weak user signals, the exponential term in \eqref{eq:pout_closed_exp} can be approximated as \(e^x\approx 1+x\). Therefore, the output of the probe beam is
\begin{equation}\small
    P_{{\rm out},m}(t)
    \approx
     P_{0,m}^{c}(\mathbf{q}_m,\mathbf{a})
    +
    \Delta P_{{\rm out},m}(t),
    \label{eq:pout_dc_ac}
\end{equation}
where the AC component is
\begin{equation}\small
\begin{aligned}
    &\Delta P_{{\rm out},m}(t)
    \approx
       k_pD_{\Omega}L P_{0,m}^{c}(\mathbf{q}_{m},\mathbf{a})
   g_m^{c}(\mathbf{q}_{m},\mathbf a)
    \sum_{k=1}^{K}
    \Omega_k\\
    &\times\sinc
    \Big[
    \frac{\kappa_{m,k}^{c}(\mathbf{q}_{m},\mathbf a)L}{2}
    \Big]
    \cos
    \left[
    2\pi f_{D,k}t
    +
    \Delta\Phi_{m,k}^{c}(\mathbf{q}_{m},\mathbf a)
    \right].
\end{aligned}
    \label{eq:pout_ac_closed}
\end{equation}

\subsection{Equivalent Baseband Channel Model}
The AC photocurrent of interest is
\begin{equation}\small
    i_{{\rm AC},m}(t)
    =
    \mathcal R_{\rm pd}
    \Delta P_{{\rm out},m}(t)
    +
    n_{i,m}(t),
    \label{eq:ac_photocurrent_rewrite}
\end{equation}
where \(\mathcal R_{\rm pd}\) is the photodetector responsivity and \(n_{i,m}(t)\) denotes the equivalent photocurrent noise. The photocurrent is converted into voltage through a unit load resistance \(R_0 = 1\) ohm. The voltage after the electrical amplifier with gain \(G_{\rm amp}\) is 
\begin{equation}\small
\begin{aligned}
    v_m(t)
     &\approx
    G_{\rm amp}R_0\mathcal R_{\rm pd}
    k_pD_{\Omega} L P_{0,m}^{c}(\mathbf{q}_{m},\mathbf a) g_m^{c}(\mathbf{q}_{m},\mathbf a)
     \\
    &\times 
    \sum_{k=1}^{K}
    \Omega_k\sinc
    \Big[
    \frac{\kappa_{m,k}^{c}(\mathbf{q}_{m},\mathbf a)L}{2}
    \Big]
   \cos
    \big[
    2\pi f_{D,k}t
   +\\&
    \Delta\Phi_{m,k}^{c}(\mathbf{q}_{m},\mathbf a)
    \big] 
    +
    n_{v,m}(t),
\end{aligned}
    \label{eq:voltage_ac_output}
\end{equation}
where \(n_{v,m}(t)=G_{\rm amp}R_0n_{i,m}(t)\) denotes the equivalent noise. Then, by mixing and filtering, the baseband signal is expressed as 
\begin{equation}\small\label{bbsignal}
\begin{split}
        y_{{\rm B}, m}(t) &= G_{\rm amp}R_0\mathcal R_{\rm pd}
    k_pD_{\Omega}L P_{0,m}^{c}(\mathbf{q}_{m},\mathbf a)
    \\
    &\times g_m^{c}(\mathbf{q}_{m},\mathbf a)\sum_{k=1}^{K}
    \Omega_k
    \sinc
    \Big[
    \frac{\kappa_{m,k}^{c}(\mathbf{q}_{m},\mathbf a)L}{2}
    \Big]\\
    &\times e^{{\rm j} [2\pi f_{D,k}t+ \Delta\Phi_{m,k}^{c}(\mathbf{q}_{m},\mathbf a)]} + n_m(t),
\end{split}
\end{equation}
where \(n_{m}(t)\) is the noise, including quantum projection noise, photon shot noise, and thermal noise \cite{Gongtcom}.

Then, we derive the equivalent channel model. For the \(k\)-th user, the
received field amplitude follows the free-space link budget, i.e., \(E_{m,k}
=
\frac{\sqrt{60P_{t,k}G_{k}}}{d_{{\rm pr},k}}|h_{m,k}s_k|\),
where \(P_{t,k}\) is the transmit power, \(G_{k}\) is the transmit antenna gain, and \(d_{{\rm pr},k}\) is the propagation distance. \(h_{m,k}\in \mathbb{C}^{1\times 1}\) and \(s_k\in \mathbb{C}^{1\times 1}\) denote the channel coefficient and transmitted symbol \footnote{Note that the phase of \(s_k\) is absorbed into \(\phi_k\). Furthermore, we assume that the \(k\)-th user transmits the narrowband symbol \(s_k(t)\) with \(\mathbb E[|s_k(t)|^2]=1\)}, respectively.
Based on these definitions, we have
\begin{equation}\small
\begin{aligned}
   y_{{\rm B}, m}(t)
&=
\sum_{k=1}^{K}
w_{{\rm ce},m,k}(\mathbf{q}_{m},\mathbf a)
w_{{\rm opt},m,k}(\mathbf{q}_{m})
h_{m,k}^{\rm cen}
\\
&\times \frac{\sqrt{60P_{t,k}G_k}}{d_{{\rm pr},k}}
s_k
e^{{\rm j}2\pi f_{D,k}t}e^{-{\rm j}\Phi_{{\rm LO},m}^{c}(\mathbf{q}_{m},\mathbf a)} +
 n_m(t),
\end{aligned}
\label{eq:complex_if_voltage_model}
\end{equation}
where \(h_{m,k}^{\rm cen}\), \(w_{{\rm opt},m,k}(\mathbf{q}_{m})\), and \(w_{{\rm ce},m,k}(\mathbf{q}_{m},\mathbf a)\) are given by
\begin{equation}\small
h_{m,k}^{\rm cen}
=
e^{{\rm j}\left(k_c\mathbf u_k^{\rm T}\mathbf c_m+\phi_k-\phi_{k,s}\right)},
\label{eq:h_center_def}
\end{equation}
\begin{equation}\small
w_{{\rm opt},m,k}(\mathbf{q}_{m})
=
e^{{\rm j}(k_c u_{y,k} y_m+k_cu_{z,k}z_m)  },
\label{eq:w_opt_def}
\end{equation}
and
\begin{equation}\small
\begin{aligned}
w_{{\rm ce},m,k}(\mathbf{q}_{m},\mathbf a)
&=
G_{\rm amp}R_0\mathcal R_{\rm pd}
k_pD_{\Omega} L
\frac{\mu_{34}}{\hbar}P_{0,m}^{c}(\mathbf{q}_{m},\mathbf a)
\\
&\times g_m^{c}(\mathbf{q}_{m},\mathbf a)
\sinc
\Big[
\frac{\kappa_{m,k}^{c}(\mathbf{q}_{m},\mathbf a)L}{2}
\Big]
.
\end{aligned}
\label{eq:w_cell_transduction}
\end{equation}

For the whole vapor cell array, define \(M=M_xM_y\), the stacked baseband
observation
\(\tilde{\mathbf y}_{\rm B}(t)=[y_{{\rm B},1}(t),\ldots,y_{{\rm B},M}(t)]^{\rm T}\),
and the optical displacement vector
\(\mathbf Q=[\mathbf{q}_1,\ldots,\mathbf{q}_M]\). Collecting
\([\mathbf H_{\rm cen}]_{m,k}=h_{m,k}^{\rm cen}\),
\([\mathbf W_{\rm opt}(\mathbf Q)]_{m,k}=w_{{\rm opt},m,k}(\mathbf{q}_{m})\), and
\([\mathbf W_{\rm ce}(\mathbf Q,\mathbf a)]_{m,k}
=w_{{\rm ce},m,k}(\mathbf{q}_{m},\mathbf a)\), the equivalent baseband model is
expressed as
\begin{equation}\small
\tilde{\mathbf y}(t)
=
   \boldsymbol{\Phi}^c_{\rm LO}(\mathbf{Q},\mathbf a)\mathbf H(\mathbf Q,\mathbf a)\mathbf{D}
\mathbf P^{1/2}
\mathbf s(t)
+
\mathbf n(t),
\label{eq:array_baseband_model}
\end{equation}
where \(  \boldsymbol{\Phi}^c_{\rm LO}(\mathbf{Q},\mathbf a) = {\rm diag}\{e^{-{\rm j}\Phi^c_{\rm LO,1}(\mathbf{q}_{1},\mathbf a)},\cdots,e^{-{\rm j}\Phi^c_{\rm LO,M}(\mathbf{q}_{M},\mathbf a)}\}\), \(\mathbf H(\mathbf Q,\mathbf a)
    =
    \mathbf W_{\rm ce}(\mathbf Q,\mathbf a)
    \odot
    \mathbf W_{\rm opt}(\mathbf Q)
    \odot
    \mathbf H_{\rm cen}\),
\(\mathbf{D} = {\rm diag}\{e^{{\rm j}2\pi f_{D,1} t},\cdots,
e^{{\rm j}2\pi f_{D,K} t}\}\),
\(\mathbf{s}(t) = [s_1(t),\cdots,s_K(t)]^{\rm T}\),
\(\mathbf{n}(t) = [n_1(t),\cdots,n_{M}(t)]^{\rm T}\), and
\(\mathbf P^{1/2}
=
\operatorname{diag}
\Big\{
\frac{\sqrt{60P_{t,1}G_1}}{d_{{\rm pr},1}},
\ldots,
\frac{\sqrt{60P_{t,K}G_K}}{d_{{\rm pr},K}}
\Big\}\).

\subsection{Analysis}
\begin{remark}[Physical interpretation of RF-to-optical conversion]\label{remark1}
The coefficient $w_{{\rm ce},m,k}(\mathbf{q}_{m},\mathbf a)$ in
\eqref{eq:complex_if_voltage_model} characterizes the end-to-end
transduction from the incident RF field to the equivalent baseband
signal and can be interpreted as the product of three factors.
First,
$\sinc\!\big[\kappa_{m,k}^{c}(\mathbf{q}_{m},\mathbf a)L/2\big]$ represents the analog phase-matching factor resulting from the integration along the vapor cell. When the LO phase gradient satisfies
$\eta_{{\rm LO},m}^{c}=k_cu_{x,k}$, the full aperture gain is achieved. Otherwise, the user signal is attenuated.
Second,
$g_m^{c}(\mathbf{q}_{m},\mathbf a)$ denotes the RF-to-optical conversion gain, determined by the slope of the dispersive EIT response $\Im\{\rho_{21}\}$ at the LO operating point and programmable through the LO bias.
Finally,
$P_{0,m}^{c}(\mathbf{q}_{m},\mathbf a)$ is the LO-biased DC probe power,
which is determined by the probe-beam power and the LO operating
point.
\end{remark}

\begin{remark}[Optical-based movable antenna]
The optical displacement \(\mathbf{q}_m\) enables each vapor cell to operate as a movable antenna without mechanical actuation. Compared with conventional MAs, the probe-coupling overlap region is reconfigured, while the vapor cell and all optical and electrical interfaces remain stationary, significantly reducing the implementation complexity. According to \eqref{eq:w_opt_def}, the resulting channel is equivalent to that of a conventional antenna physically translated to \(\mathbf{c}_m + \mathbf{q}_m\). Therefore, the proposed architecture inherits the channel-conditioning capability of movable/fluid antennas while avoiding physical translation of the vapor cell and its RF/optical interfaces, at the expense of an additional optical beam-steering subsystem.
\end{remark}

\begin{remark}[Multiuser decorrelation via optical displacement]
\label{prop:decorrelation_2d}
Define the multiuser Gram matrix as
\begin{equation}
\widehat{\mathbf R}(\mathbf Q)
=
\mathbf H^{\rm H}(\mathbf Q,\mathbf a)
\mathbf H(\mathbf Q,\mathbf a),
\end{equation}
whose $(i,j)$-th entry is
\begin{equation}\small
\begin{split}
\hat{\varrho}_{ij}(\mathbf Q)
&=
\sum_{m=1}^{M}
w_{{\rm ce},m,i}^{*}(\mathbf q_m,\mathbf a)
w_{{\rm ce},m,j}(\mathbf q_m,\mathbf a)
\left(h^{\rm cen}_{m,i}\right)^{*}h^{\rm cen}_{m,j}\\
&\quad\times
e^{-{\rm j}k_c
\Delta\mathbf u_{\perp,i,j}^{\rm T}\mathbf q_m}.
\end{split}
\label{eq:column_correlation_2d}
\end{equation}
$\Delta\mathbf u_{\perp,i,j}
=[0,u_{y,i}-u_{y,j},u_{z,i}-u_{z,j}]^{\rm T}$.
Perfect decorrelation among all $K$ users is achieved when for all
$\hat{\varrho}_{ij}(\mathbf Q)=0$, $i\neq j$, or
equivalently, when $\widehat{\mathbf R}(\mathbf Q)$ is diagonal. Accordingly, the following design insights can be obtained.
\begin{enumerate}
\item \textit{Larger vapor-cell radii enhance user decorrelation:}  Since $|\mathbf q_m|^2\leq R_{\rm VC}$, the total geometric
phase range available for user pair $(i,j)$ is \(
2k_cR_{\rm VC}
|\Delta\mathbf u_{\perp,i,j}|2\). Therefore,
a larger $R_{\rm VC}$  provides more controllable phase tuning and enhances spatial reconfigurability through the position-dependent conversion gain.

\item \textit{Conditions for perfect multiuser decorrelation:} Perfect decorrelation requires the $K$ nonzero user channels to be mutually orthogonal.
Since $\mathbf H(\mathbf Q,\mathbf a)\in\mathbb C^{M\times K}$,
a necessary condition is $M\geq K$. Otherwise, perfect decorrelation is impossible regardless of the optical position and LO designs. Meanwhile, even though $M\geq K$, perfect decorrelation can be achieved only if the users exhibit sufficiently distinguishable position-dependent responses and there exists a common feasible design $(\mathbf Q,\mathbf a)$ that makes their equivalent channels mutually orthogonal.
\item \textit{Limitations and practical design implications:} Although increasing \(M\) enlarges the receive dimension and improves multiuser separability, transverse optical displacement introduces no relative phase between users \(i\) and \(j\) when \(\Delta\mathbf u_{\perp,i,j}=\mathbf 0\). In this case, phase-domain discrimination is unavailable. However, gain-domain discrimination may still be effective through \(\operatorname{sinc}\left[
\frac{(k_cu_{x,k}-\eta_{{\rm LO},m}^{c})L}{2}
\right]\). For same-half-space users with identical arrival directions, both \(u_{x,i}=u_{x,j}\) and \(\Delta\mathbf u_{\perp,i,j}=\mathbf 0\) hold. In this case, neither optical-displacement-induced phase control nor phase-matching gain can distinguish the two users, and increasing \(M\) cannot improve their spatial separability.
\end{enumerate}
\end{remark}

\section{Optical and LO Design}

In this section, we employ a linear MRC detector in the receiver and jointly design the near-field LO and the optical positions to maximize the sum rate. 

\subsection{Problem Formulation}
\label{subsec:mrc_sum_rate}
We focus on the joint design of the near-field LO and the optical positions under the assumptions of perfect channel state information and perfect Doppler compensation. Furthermore, the user's transmit power is fixed and thus is omitted from the
subsequent optimization. Accordingly, the compensated baseband signal
is expressed as
\begin{equation}\small
\widetilde{\mathbf y}(t)
=
\boldsymbol{\Phi}_{\rm LO}^{c}(\mathbf Q,\mathbf a)
\mathbf H(\mathbf Q,\mathbf a)
\mathbf s(t)
+
\mathbf n(t).
\label{eq:mse_received_model}
\end{equation}
Generally, the receiver noise consists of multiple components,
including electronic thermal noise, photon shot noise, and quantum
projection noise 
\cite{Gongtcom}. Although the photon shot noise and quantum projection noise generally depend on the LO and optical-position designs, the direct intensity of detection architecture considered here operates in a thermal-noise-dominated regime \cite{Gongtcom}. Therefore, the variation of the total noise power with \(\mathbf Q\) and \(\mathbf a\) is negligible compared with the position-independent thermal-noise floor. Accordingly, we approximate the noise covariance as \(\mathbf n(t)
\sim
\mathcal{CN}
\left(
\mathbf 0,\mathbf R_n=\sigma^2\mathbf I_M
\right)\).
With an MRC detector \(\mathbf{W}_{\rm MRC}(\mathbf Q,\mathbf a) = \boldsymbol\Phi^{c}_{\rm LO}(\mathbf Q,\mathbf a)\mathbf H(\mathbf Q,\mathbf a) \), the received signal from the \(k\)-th user is
\begin{equation}\small
\begin{split}
        \hat s_k(t)
    &=
    \mathbf w_k^{\rm H}(\mathbf Q,\mathbf a)
    \widetilde{\mathbf y}(t)
    =
       \hat \varrho_{kk}(\mathbf Q,\mathbf a)\,s_k(t)
    \\
    &+
    \sum_{j\neq k}
        \hat\varrho_{kj}(\mathbf Q,\mathbf a)\,s_j(t)
    +
    \tilde n_k(t),
\end{split}
    \label{eq:mrc_output}
\end{equation}
where \(\mathbf{w}_{k}(\mathbf Q,\mathbf a)\) is the \(k\)-th column  of \(\mathbf{W}_{\rm MRC}\) and $\tilde n_k(t)
={\mathbf h}_k^{\rm H}
[\boldsymbol\Phi^{c}_{\rm LO}]^{\rm H}\mathbf n(t)$ has variance
of $\sigma^2\hat\varrho_{kk}(\mathbf Q,\mathbf a)$. \( \hat\varrho_{kj}(\mathbf Q,\mathbf a) ={\mathbf h}_k^{\rm H}(\mathbf Q,\mathbf a){\mathbf h}_j(\mathbf Q,\mathbf a) \), where \({\mathbf h}_k(\mathbf Q,\mathbf a)\) denotes the \(k\)-th column of \(\mathbf H(\mathbf Q,\mathbf a)\).  With
$\mathbb E[\mathbf s(t)\mathbf s^{\rm H}(t)]=\mathbf I_K$, the
received SINR of user $k$ is
\begin{equation}\small
    {\rm SINR}_k(\mathbf Q,\mathbf a)
    =
    \frac{    \hat \varrho_{kk}^2(\mathbf Q,\mathbf a)}{\Upsilon_k(\mathbf Q,\mathbf a)},
    \label{eq:mrc_sinr}
\end{equation}
where \( \Upsilon_k(\mathbf Q,\mathbf a)
    \triangleq
    \sum_{j\neq k}
    |    \hat \varrho_{kj}(\mathbf Q,\mathbf a)|^2
    +
    \sigma^2    \hat\varrho_{kk}(\mathbf Q,\mathbf a)\).
As can be seen, the effective energy $    \hat \varrho_{kk}(\mathbf Q,\mathbf a)$ and user interference $|    \hat\varrho_{kj}(\mathbf Q,\mathbf a)|^2$ depend on the design variables exclusively through user correlation. Therefore, the MRC criterion directly rewards the two
mechanisms quantified in Remark~\ref{prop:decorrelation_2d},
namely enlarging the per-user captured energy and reducing the
inter-user correlations.

Based on \eqref{eq:mrc_sinr}, the achievable sum rate
under MRC is $R_{\rm MRC}(\mathbf Q,\mathbf a)
=\sum_{k=1}^{K}\log_2[1+{\rm SINR}_k(\mathbf Q,\mathbf a)]$. The sum rate maximization is formulated as
\begin{subequations}\small
\begin{align}
\mathcal P_1:\quad
\max_{\mathbf Q,\mathbf a}\quad&
R_{\rm MRC}(\mathbf Q,\mathbf a)
\notag
\\
{\rm s.t.}\quad&
y_m^2+z_m^2\le R_{\rm VC}^2,
\quad m=1,\ldots,M,
\label{eq:mrc_position_constraint}
\\
&
\|\mathbf a\|_2^2\le1.
\label{eq:mrc_lo_constraint}
\end{align}
\label{eq:mrc_optimization}
\end{subequations}
The first constraint ensures that the $m$-th optical sensing line
remains inside the circular cross section of the corresponding vapor
cell, whereas the second constraint limits the normalized excitation
power of the near-field LO array. 

\subsection{Proposed Joint Design}
Problem~$\mathcal P_1$ is challenging to solve as it inherits the coupled nonconvex structure of the fractional SINR and sum rate. To tackle this issue, we employ an alternating algorithm to iteratively optimize the LO and optical design. Furthermore, both optimizations are solved by projected-gradient ascent. In the following, we derive the gradient of $R_{\rm MRC}(\mathbf Q,\mathbf a)$ with respect to any real design variable $\varkappa\in\{y_m,z_m\}\cup\{\beta_p,\varphi_p\}$, \(\forall m \in \{1,\cdots,M\}\), \(\forall p \in \{1,\cdots,P\}\).



First, we define the positive scalars, which are given by
\begin{equation}\small
\begin{split}
\Xi_k(\mathbf Q,\mathbf a)
&\triangleq
\frac{1}
{\Upsilon_k(\mathbf Q,\mathbf a)
\big[
\Upsilon_k(\mathbf Q,\mathbf a)
+    \hat \varrho_{kk}^2(\mathbf Q,\mathbf a)
\big]\ln2},
\\
S_k(\mathbf Q,\mathbf a)
&\triangleq
2\sum_{j\neq k}
|    \hat \varrho_{kj}(\mathbf Q,\mathbf a)|^2
+
\sigma^2    \hat\varrho_{kk}(\mathbf Q,\mathbf a).
\end{split}
\label{eq:mrc_scalar_defs}
\end{equation}
Then, we have the following lemma based on the objective function.

\begin{lemma}[MRC sum-rate gradient]
\label{lem:mrc_gradient}
Construct the Hermitian weight matrix
$\boldsymbol\Psi(\mathbf Q,\mathbf a)\in\mathbb C^{K\times K}$ with
\begin{equation}\small
\begin{split}
{[\boldsymbol\Psi(\mathbf Q,\mathbf a)]}_{i,i}
&=
2\Xi_i(\mathbf Q,\mathbf a)
    \hat\varrho_{ii}(\mathbf Q,\mathbf a)
S_i(\mathbf Q,\mathbf a),
\\
{[\boldsymbol\Psi(\mathbf Q,\mathbf a)]}_{j,i}
&=
-2
\Big[
\Xi_i(\mathbf Q,\mathbf a)
    \hat\varrho_{ii}^2(\mathbf Q,\mathbf a)
\\
&+
\Xi_j(\mathbf Q,\mathbf a)
    \hat\varrho_{jj}^2(\mathbf Q,\mathbf a)
\Big]
    \hat\varrho_{ij}^{*}(\mathbf Q,\mathbf a),
\quad j\neq i.
\end{split}
\label{eq:mrc_weight_matrix}
\end{equation}
Then, for any real design variable $\varkappa\in\{y_m,z_m\}\cup\{\beta_p,\varphi_p\}$, the gradient of the
MRC sum rate is
\begin{equation}\small
\begin{aligned}
\frac{\partial R_{\rm MRC}(\mathbf Q,\mathbf a)}{\partial\varkappa}
&=
\Re\Big\{
\sum_{m=1}^{M}
\sum_{k=1}^{K}
\big[
\boldsymbol\Lambda(\mathbf Q,\mathbf a)
\big]_{m,k}^{*}
\\
&\times
\frac{\partial[\mathbf H(\mathbf Q,\mathbf a)]_{m,k}}
{\partial\varkappa}
\Big\},
\end{aligned}
\label{eq:mrc_master_gradient}
\end{equation}
where $\boldsymbol\Lambda(\mathbf Q,\mathbf a)
\triangleq
\mathbf H(\mathbf Q,\mathbf a)
\boldsymbol\Psi(\mathbf Q,\mathbf a)$.
\end{lemma}

\begin{IEEEproof}
See Appendix~\ref{app:mrc_gradient_proof}.
\end{IEEEproof}

Based on Lemma \ref{lem:mrc_gradient}, the gradient of the sum rate can be obtained by using the partial derivative of the equivalent channel. Therefore,  we iteratively optimize the optical and LO design by deriving their corresponding derivatives.
\subsubsection{Optical Design}
For the given LO design at the \(r\)-th iteration, the gradient of the sum rate is given by
\begin{equation}\small
\begin{split}
    \frac{\partial R_{\rm MRC}(\mathbf Q,\mathbf a^{(r)})}{\partial \varkappa}
&=
\Re\Big\{\sum_{m=1}^{M}
\sum_{k=1}^{K}
\big[
\boldsymbol\Lambda(\mathbf Q,\mathbf a^{(r)})
\big]_{m,k}^{*}\\
&\times 
\frac{\partial[\mathbf H(\mathbf Q,\mathbf a^{(r)})]_{m,k}}
{\partial \varkappa}
\Big\},
\end{split}
\label{eq:mrc_position_gradient}
\end{equation}
where \(\varkappa\in\{y_m,z_m\}\), \(\forall m\in \{1,\cdots,M\}\). With \( [{\mathbf H}(\mathbf Q,\mathbf a^{(r)})]_{m,k} =  w_{{\rm ce},m,k}(\mathbf{q}_{m},\mathbf a^{(r)} ) w_{{\rm opt},m,k}(\mathbf{q}_{m}) h_{m,k}^{\rm cen} \), we need to derive the partial derivative of \(w_{{\rm ce},m,k}(\mathbf{q}_{m},\mathbf a^{(r)} )\) and \(w_{{\rm opt},m,k}(\mathbf{q}_{m})\).

By applying the product rule to
$[{\mathbf H}(\mathbf Q,\mathbf a^{(r)})]_{m,k}$,
the channel partial derivative is
\begin{equation}\small
\begin{aligned}
&\frac{
\partial[{\mathbf H}(\mathbf Q,\mathbf a^{(r)})]_{m,k}
}{
\partial \varkappa
}
=
h_{m,k}^{\rm cen}
e^{{\rm j}k_c
\left(
u_{y,k}y_m+u_{z,k}z_m
\right)}
\Big\{
\varpi_m
\Big[
{\rm j}k_cu_{\varkappa,k}\\
&\times
\sinc
\big[
\tfrac{\kappa_{m,k}^{c}(\mathbf{q}_{m},\mathbf a^{(r)})L}{2}
\big]-
\frac{L}{2}
\sinc'
\big[
\tfrac{\kappa_{m,k}^{c}(\mathbf{q}_{m},\mathbf a^{(r)})L}{2}
\big]\\
&\times
\frac{\partial\eta_{{\rm LO},m}^{c}(\mathbf{q}_{m},\mathbf a^{(r)})}{\partial \varkappa}
\Big]+
G_{\rm amp}R_0\mathcal R_{\rm pd}
k_pD_{\Omega}L
\frac{\mu_{34}}{\hbar}\\
&\times
P_{0,m}^{c}(\mathbf{q}_{m},\mathbf a^{(r)})
\Big[
k_pD_{\Omega}L
\big[g_m^{c}(\mathbf{q}_{m},\mathbf a^{(r)})\big]^2+
g_m^{c\prime}(\mathbf{q}_{m},\mathbf a^{(r)})
\Big]\\
&\times
\sinc
\big[
\tfrac{\kappa_{m,k}^{c}(\mathbf{q}_{m},\mathbf a^{(r)})L}{2}
\big]
\frac{
\partial\Omega_{{\rm LO},m}^{c}(\mathbf{q}_{m},\mathbf a^{(r)})
}{
\partial \varkappa
}
\Big\},
\end{aligned}
\label{eq:opt_channel_derivative}
\end{equation}
where 
\begin{equation}\small
\begin{split}
\varpi_m
&\triangleq
G_{\rm amp}R_0\mathcal R_{\rm pd}
k_pD_{\Omega}L
\frac{\mu_{34}}{\hbar}
P_{0,m}^{c}
\big(
\mathbf q_{m},\mathbf a^{(r)}
\big)
g_m^{c}
\big(
\mathbf q_{m},\mathbf a^{(r)}
\big).
\end{split}
\label{eq:opt_gamma}
\end{equation}
The remaining terms can be found in Appendix \ref{app:opt_gradient_proof}.

Based on the derived gradient, the optical positions are updated
sequentially by a projected-gradient method, as summarized in
Algorithm~\ref{alg:mrc_optical_ascent}.

\begin{algorithm}[t]
\caption{Projected-Gradient-Ascent Optical Design for MRC Sum-Rate
Maximization}
\label{alg:mrc_optical_ascent}
\small
\begin{algorithmic}[1]
\STATE Initialize $\ell=0$, the optical positions
$\mathbf Q^{(r)}_{0}=\mathbf Q^{(r-1)}$, the initial step size
$\epsilon$, the step-reduction factor $0<\tau_{\epsilon}<1$, and the
stopping tolerance $\varepsilon_{\rm stop}$.
\STATE Compute
$R^{(r)}_{0}=R_{\rm MRC}\big(\mathbf Q^{(r)}_{0},\mathbf a^{(r)}\big)$,
and set the relative increase $\Delta R=+\infty$.
\WHILE{$\Delta R>\varepsilon_{\rm stop}$}
\STATE Update $\ell\leftarrow\ell+1$ and set
$\epsilon_{\ell}=\epsilon$;
\STATE Calculate the full ascent direction
$\mathbf G^{(r)}_{\ell}
=\nabla_{\mathbf Q}R_{\rm MRC}
\big(\mathbf Q^{(r)}_{\ell-1},\mathbf a^{(r)}\big)$;
\STATE Calculate the tentative positions
$\mathbf Q_{\rm ten}
=\Pi_{\mathcal Q}
\big(\mathbf Q^{(r)}_{\ell-1}
+\epsilon_{\ell}\mathbf G^{(r)}_{\ell}\big)$, and evaluate the tentative sum rate
$R_{\rm ten}
=R_{\rm MRC}\big(\mathbf Q_{\rm ten},\mathbf a^{(r)}\big)$;
\WHILE{$R_{\rm ten}\le R^{(r)}_{\ell-1}$}
\STATE  Update the step size
$\epsilon_{\ell}\leftarrow\tau_{\epsilon}\epsilon_{\ell}$, recompute the tentative positions
$\mathbf Q_{\rm ten}
=\Pi_{\mathcal Q}
\big(\mathbf Q^{(r)}_{\ell-1}
+\epsilon_{\ell}\mathbf G^{(r)}_{\ell}\big)$ and the tentative
objective
$R_{\rm ten}
=R_{\rm MRC}\big(\mathbf Q_{\rm ten},\mathbf a^{(r)}\big)$;
\ENDWHILE
\STATE Update the positions
$\mathbf Q^{(r)}_{\ell}=\mathbf Q_{\rm ten}$, the objective
$R^{(r)}_{\ell}=R_{\rm ten}$, and the relative increase
$\Delta R=\big(R^{(r)}_{\ell}-R^{(r)}_{\ell-1}\big)\big/
R^{(r)}_{\ell}$;
\ENDWHILE
\STATE Set the updated optical-position matrix as
$\mathbf Q^{(r)}\leftarrow\mathbf Q^{(r)}_{\ell}$.
\end{algorithmic}
\end{algorithm}

\subsubsection{LO Design}
Similarly, with the given optical design $\mathbf Q^{(r)}$ at the \(r\)-th iteration, we have
\begin{equation}\small
\frac{
\partial
[\mathbf H(\mathbf Q^{(r)},\mathbf a)]_{m,k}
}{
\partial \varkappa
}
=
w_{{\rm opt},m,k}(\mathbf q_m^{(r)})
h_{m,k}^{\rm cen}
\frac{
\partial w_{{\rm ce},m,k}
(\mathbf q_m^{(r)},\mathbf a)
}{
\partial \varkappa
},
\label{eq:LO_channel_derivative_decomposition}
\end{equation}
where \(\varkappa \in \{\beta_p,\varphi_p\}\), \(\forall p \in \{1,\cdots,P\}\).
With the expression given in \eqref{eq:w_cell_transduction}, we have
\begin{equation}\small
\begin{aligned}
\frac{
\partial w_{{\rm ce},m,k}
}{
\partial \varkappa
}
&=
G_{\rm amp}R_0\mathcal R_{\rm pd}
k_pD_{\Omega}L\frac{\mu_{34}}{\hbar}
P_{0,m}^{c}(\mathbf{q}^{(r)}_{m},\mathbf{a})
\\
&\times
\Bigg\{
\sinc
\Big[
\frac{\kappa_{m,k}^{c}(\mathbf{q}^{(r)}_{m},\mathbf{a})L}{2}
\Big]
\Big[
k_pD_{\Omega}L
\big(g_m^{c}(\mathbf{q}^{(r)}_{m},\mathbf{a})\big)^2
\\
&+ g_m^{c\prime}(\mathbf{q}^{(r)}_{m},\mathbf{a})
\Big]
\frac{
\partial\Omega_{{\rm LO},m}^{c}(\mathbf{q}^{(r)}_{m},\mathbf{a})
}{
\partial \varkappa
}
-
\frac{L}{2}
g_m^{c}(\mathbf{q}^{(r)}_{m},\mathbf{a})\\
&\times 
\sinc'
\Big[
\frac{\kappa_{m,k}^{c}(\mathbf{q}^{(r)}_{m},\mathbf{a})L}{2}
\Big]
\frac{
\partial\eta_{{\rm LO},m}^{c}(\mathbf{q}^{(r)}_{m},\mathbf{a})
}{
\partial \varkappa
}
\Bigg\},
\end{aligned}
\label{eq:LO_cell_transduction_derivative}
\end{equation}
where \(g_m^{c\prime}(\mathbf{q}^{(r)}_{m},\mathbf{a})\), \(\frac{
\partial\Omega_{{\rm LO},m}^{c}(\mathbf{q}^{(r)}_{m},\mathbf{a})
}{
\partial \varkappa
}\), and \(\frac{
\partial\eta_{{\rm LO},m}^{c}(\mathbf{q}^{(r)}_{m},\mathbf{a})
}{
\partial \varkappa
}\) can be found in Appendix \ref{app:LO_gradient_derivation}.

Similarly, with the given gradient, we can adopt the projected gradient ascent algorithm to maximize the objective function, as detailed in Algorithm \ref{alg:mrc_lo_ascent}.

\begin{algorithm}[t]
\caption{Projected-Gradient-Ascent LO Design for MRC Sum-Rate
Maximization}
\label{alg:mrc_lo_ascent}
\small
\begin{algorithmic}[1]
\STATE Initialize $\ell=0$, the LO amplitude and phase vectors
$\boldsymbol\beta^{(r)}_{0}=\boldsymbol\beta^{(r-1)}$ and
$\boldsymbol\varphi^{(r)}_{0}=\boldsymbol\varphi^{(r-1)}$, namely
$\mathbf a^{(r)}_{0}
=\boldsymbol\beta^{(r)}_{0}\odot
e^{{\rm j}\boldsymbol\varphi^{(r)}_{0}}
=\mathbf a^{(r-1)}$, the initial step size $\epsilon$, the
step-reduction factor $0<\tau_{\epsilon}<1$, and the stopping
tolerance $\varepsilon_{\rm stop}$.
\STATE Compute
$R^{(r)}_{0}
=R_{\rm MRC}\big(\mathbf Q^{(r)},\mathbf a^{(r)}_{0}\big)$, and set
the relative increase $\Delta R=+\infty$.
\WHILE{$\Delta R>\varepsilon_{\rm stop}$}
\STATE Update $\ell\leftarrow\ell+1$ and set
$\epsilon_{\ell}=\epsilon$;
\STATE Calculate the ascent directions
$\mathbf g^{(r)}_{\beta,\ell}
=\nabla_{\boldsymbol\beta}R_{\rm MRC}
\big(\mathbf Q^{(r)},\mathbf a^{(r)}_{\ell-1}\big)$ and
$\mathbf g^{(r)}_{\varphi,\ell}
=\nabla_{\boldsymbol\varphi}R_{\rm MRC}
\big(\mathbf Q^{(r)},\mathbf a^{(r)}_{\ell-1}\big)$;
\STATE Calculate the tentative LO amplitude and phase vectors
$\boldsymbol\beta_{\rm ten}
=\Pi_{\mathcal B}
\big(\boldsymbol\beta^{(r)}_{\ell-1}
+\epsilon_{\ell}\mathbf g^{(r)}_{\beta,\ell}\big)$ and
$\boldsymbol\varphi_{\rm ten}
=\boldsymbol\varphi^{(r)}_{\ell-1}
+\epsilon_{\ell}\mathbf g^{(r)}_{\varphi,\ell}$, where
$\Pi_{\mathcal B}(\mathbf x)
=[\mathbf x]_{+}/\max\{1,\|[\mathbf x]_{+}\|_2\}$; then construct
$\mathbf a_{\rm ten}
=\boldsymbol\beta_{\rm ten}\odot
e^{{\rm j}\boldsymbol\varphi_{\rm ten}}$ and evaluate the tentative
objective
$R_{\rm ten}
=R_{\rm MRC}\big(\mathbf Q^{(r)},\mathbf a_{\rm ten}\big)$;
\WHILE{$R_{\rm ten}\le R^{(r)}_{\ell-1}$}
\STATE Update the step size
$\epsilon_{\ell}\leftarrow\tau_{\epsilon}\epsilon_{\ell}$ and recompute
$\boldsymbol\beta_{\rm ten}$,
$\boldsymbol\varphi_{\rm ten}$,
$\mathbf a_{\rm ten}$, and $R_{\rm ten}$ as in Step 6;
\ENDWHILE
\STATE Update the LO vectors
$\boldsymbol\beta^{(r)}_{\ell}=\boldsymbol\beta_{\rm ten}$,
$\boldsymbol\varphi^{(r)}_{\ell}=\boldsymbol\varphi_{\rm ten}$, and
$\mathbf a^{(r)}_{\ell}=\mathbf a_{\rm ten}$, the objective
$R^{(r)}_{\ell}=R_{\rm ten}$, and the relative increase
$\Delta R=\big(R^{(r)}_{\ell}-R^{(r)}_{\ell-1}\big)\big/
R^{(r)}_{\ell}$;
\ENDWHILE
\STATE Set the updated LO excitation vector as
$\mathbf a^{(r)}\leftarrow\mathbf a^{(r)}_{\ell}$.
\end{algorithmic}
\end{algorithm}

\subsubsection{AO Optimization} Given the updated optical positions and LO excitation, the MRC combiner
is updated as
$\mathbf W_{\rm MRC}(\mathbf Q,\mathbf a)
=\boldsymbol\Phi^{c}_{\rm LO}(\mathbf Q,\mathbf a)
\mathbf H(\mathbf Q,\mathbf a)$.
Algorithm~\ref{alg:mrc_optical_ascent} and
Algorithm~\ref{alg:mrc_lo_ascent} are executed alternately until the
sum rate converges.

\subsection{Convergence and Complexity Analysis}

The projection operation guarantees that every updated optical position and LO design remains feasible. Moreover, the backtracking procedure is adopted to ensure a sufficient increase in the MRC sum rate. Consequently, the generated objective sequence is monotonically non-decreasing, ensuring its stable convergence.

As for the complexity of our proposed algorithm, it mainly depends on the number of iterations and the complexity of each block. 
For the optical block, computing the position derivatives of the
superimposed LO and the equivalent channels requires
$\mathcal O(MP+MK^2)$ operations. For the LO block, evaluating the closed-form gradients with respect to the amplitudes and phases of all $P$ LO elements requires $\mathcal O(MPK+MK^2)$ operations. Consequently, the overall complexity of the proposed alternating
optimization algorithm is
\(
\mathcal O\!\Big(
I_{\rm AO}[I_{\rm LO}(MPK+MK^2)
+
I_{\rm Q}(MP+MK^2)]
\Big)
\), where \(I_{\rm AO}\), \(I_{\rm LO}\), and \(I_{\rm Q}\) denote the number of iterations in alternating optimization, LO block, and optical block, respectively.

\section{Simulation Results}
In this section, we first validate our derived closed-form expressions and then evaluate the feasibility of the optically movable antenna by designing the laser coupling position. Furthermore, we evaluate the effectiveness of the proposed algorithm. 
\subsection{Parameter Settings}
Consider a four-level excitation scheme 
\(6{\rm S}_{1/2}\rightarrow 6{\rm P}_{3/2}\rightarrow 47{\rm D}_{5/2}\rightarrow 48{\rm P}_{3/2}\). 
Unless otherwise stated, the detailed parameter settings are consistent with those in \cite{peng2026celllevelchannelshapingrydberg}. The receiver employs an 
\(M_x\times M_y = 4\times4\) array of vapor cells, each of length \(L = 4\) cm, 
with inter-cell spacings along the \(x\)- and \(y\)-axes set to 
\(d_x = d_y = \lambda_c/2\). The radius of each vapor cell is \(R_{\rm VC} = \lambda_c/2\). The RF channel of user \(k\) is modeled as a Rician 
channel, given by
\[
\mathbf{h}_{{\rm RF},k} = \sqrt{\tfrac{R_{{\rm R},k} L_{{\rm loss},k}}{R_{{\rm R},k}+1}}\,
\bar{\mathbf{h}}_{{\rm RF},k} 
+\sqrt{\tfrac{L_{{\rm loss},k}}{R_{{\rm R},k}+1}}\,
\tilde{\mathbf{h}}_{{\rm RF},k} \in \mathbb{C}^{M \times 1},
\]
where \(R_{{\rm R},k}\) is the Rician factor and 
\(L_{{\rm loss},k} = P_{t,k}G_k\left(\tfrac{\lambda_c}{4\pi d_{{\rm pr},k}}\right)^2\) 
denotes the large-scale fading of user \(k\). \(d_{{\rm pr},k}\) denotes the distance between the user \(k\) and the RAQR. \(\bar{\mathbf{h}}_{{\rm RF},k}\) is the deterministic LoS steering vector determined by the angle of arrival of user \(k\), and \(\tilde{\mathbf{h}}_{{\rm RF},k}\sim\mathcal{CN}(\mathbf{0},\mathbf{I}_{M})\) is the NLoS component, which is independent of \(\mathbf{q}_m\), \(\forall m\). Similarly, the electric-field channel is 
expressed as
\[
\mathbf{h}_k = E_k\left(\sqrt{\tfrac{R_{{\rm R},k}}{R_{{\rm R},k}+1}}\,
\bar{\mathbf{h}}_{{\rm RF},k} 
+\sqrt{\tfrac{1}{R_{{\rm R},k}+1}}\,
\tilde{\mathbf{h}}_{{\rm RF},k}\right) \in \mathbb{C}^{M \times 1},
\]
where \(E_k\) denotes the electric-field amplitude from the user \(k\) to the RAQR.

\subsection{Validation of Derived Results}
Fig.~\ref{MSEout} compares the MSE between the derived closed-form expressions and numerical results. As shown in Fig.~\ref{MSEout}(a), although the approximation error increases toward the vapor cell's boundary, all approximation errors remain below -20 dB, validating the accuracy of the center approximation. Furthermore, Fig.~\ref{MSEout}(b) presents the AC voltage outputs of different vapor cells, which further confirms the correctness of the derived analytical expressions.
\begin{figure*}[t]
	\centering
	\subfigure[MSE of center-based approximation.]{
		\includegraphics[width=3.25in]{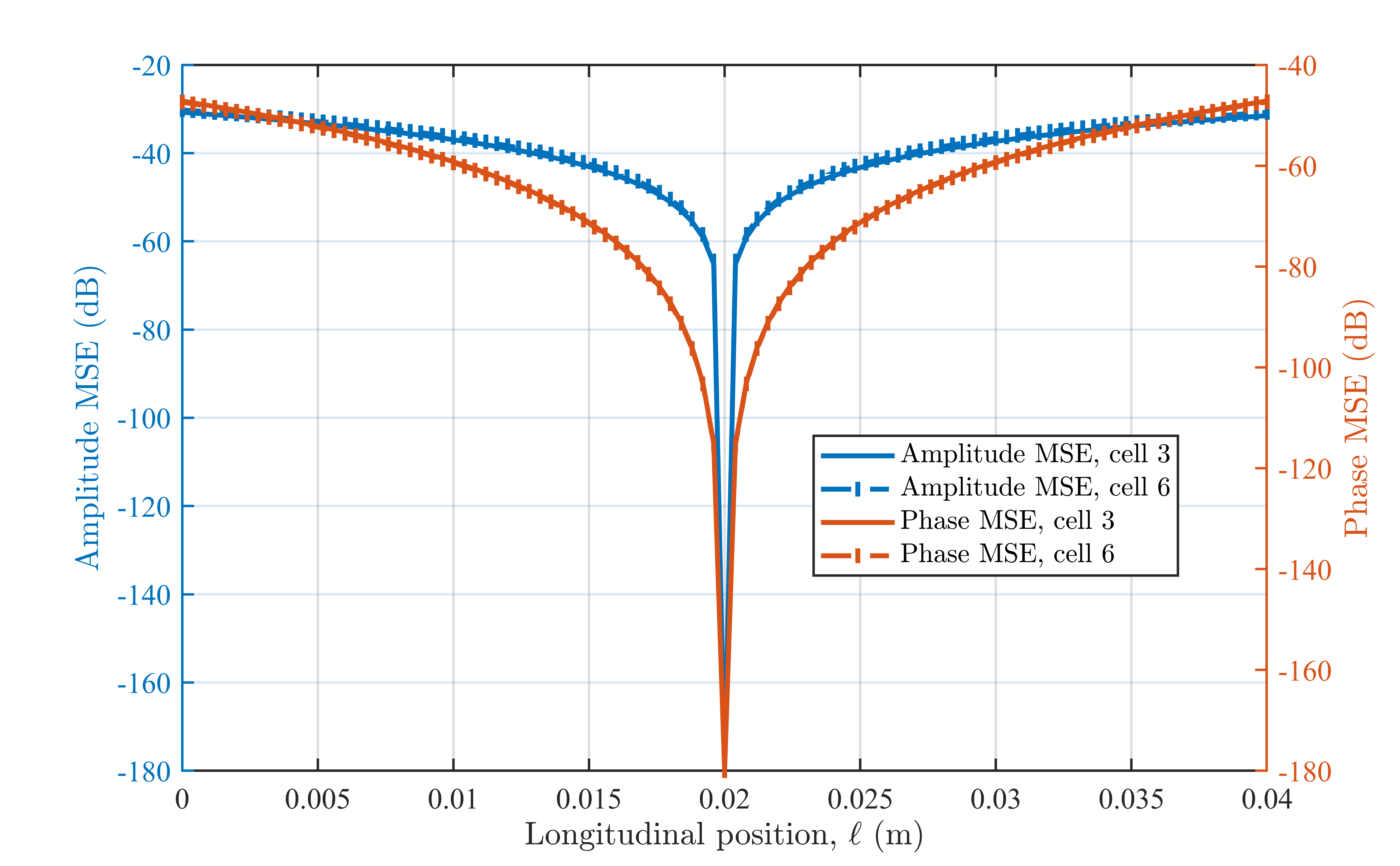}}\hspace{10mm}
	\subfigure[Output voltage.]{
		\includegraphics[width=3.25in]{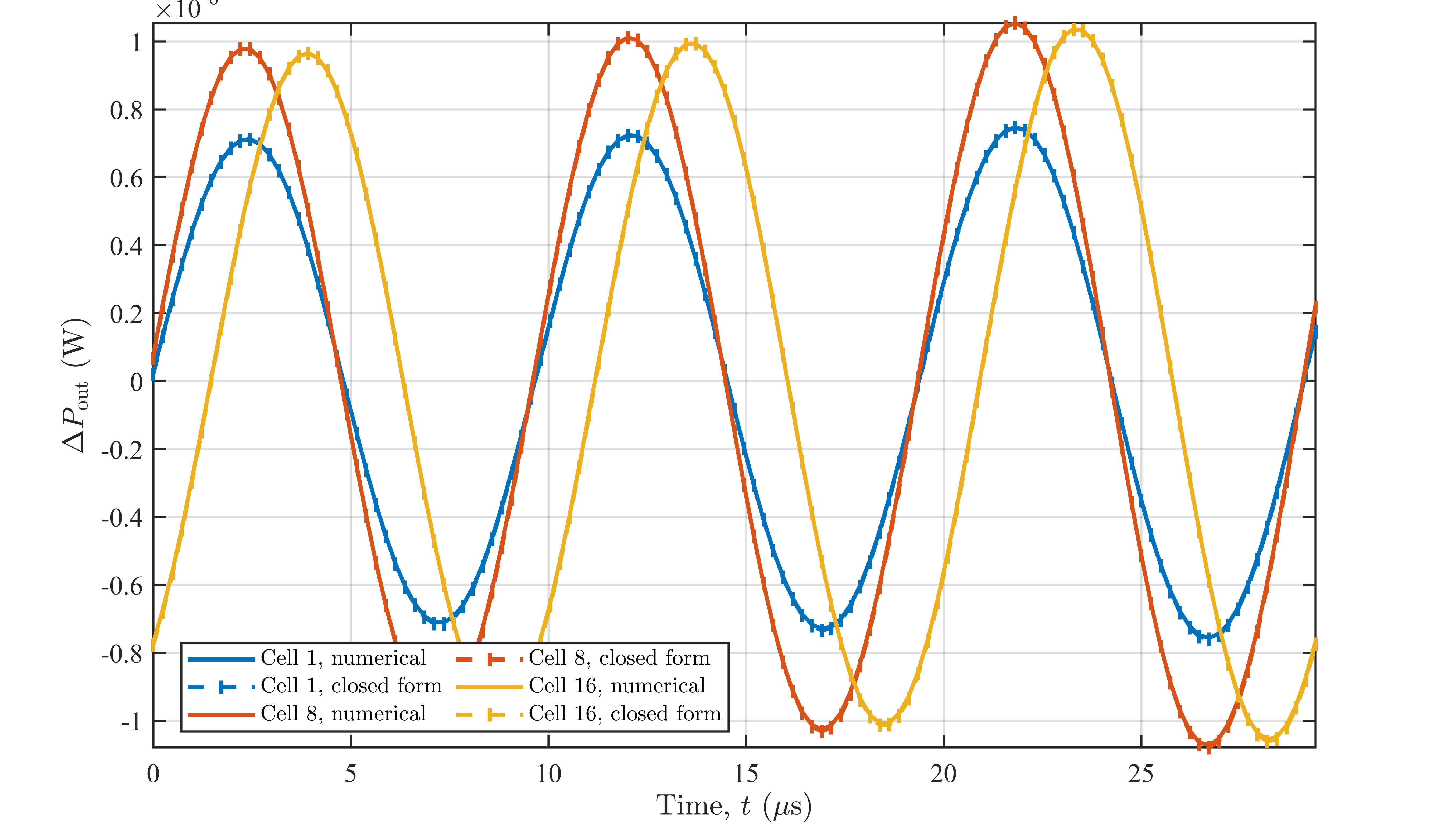}}
	\caption{Mean square error based on numerical results and derived expressions.}
    \label{MSEout}
    \vspace{-0.5cm}
\end{figure*}

\subsection{Capability of Movable RAQR}
We evaluate the benefits of the proposed LO and optical-position designs. With the fixed optical position, \(\mathbf{q}_m = \mathbf{0}\), Fig. \ref{Capabilities}(a) and Fig. \ref{Capabilities}(b) show the normalized array gain over the angular grid $(\phi,\theta)\in[-180^\circ,180^\circ]\times[0^\circ,90^\circ]$, where $\theta$ and $\phi$ denote the zenith and azimuth angles, respectively. Here, $\mathbf u_{\rm T}$ and $\mathbf u_{\rm I}$ represent the directions of the target and interference users. Compared with the conventional array, the RF-to-optical conversion suppresses the interference level from \(-19.1\) dB to \(-33.3\) dB, confirming Remark \ref{remark1}. Next, with the LO design fixed, Fig.~\ref{Capabilities}(c) presents the beam pattern obtained by optimizing the optical positions. The interference is further reduced by approximately 87 dB, demonstrating the strong beamforming capability enabled by optical displacement. Finally, Fig.~\ref{Capabilities}(d) compares the normalized channel correlation before and after RF-to-optical conversion. The significant reduction in correlation validates Remark~\ref{prop:decorrelation_2d} and confirms the capability of optically movable RAQR.

\begin{figure*}
    \centering
    \includegraphics[width=1\linewidth]{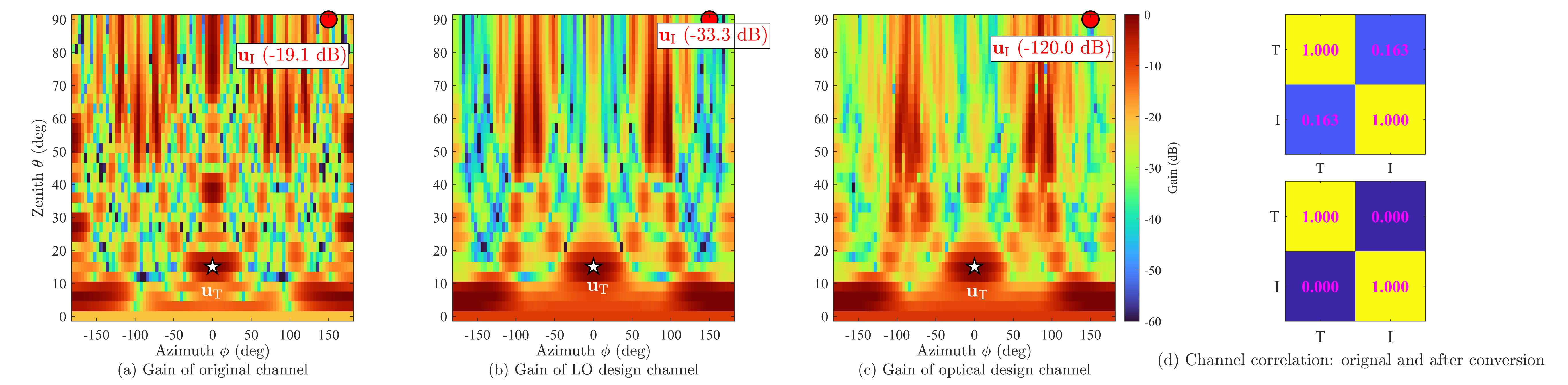}
    \caption{Capabilities of movable RAQR.}
    \label{Capabilities}
     \vspace{-0.5cm}
\end{figure*}

\subsection{Performance Evaluation of Our Algorithm}
We evaluate the convergence of our proposed algorithm in Fig.~\ref{convergencr}. As observed, the proposed algorithm reaches a stationary operating point within 4–5 outer AO iterations, corresponding to fewer than 500 iterations in total, demonstrating the rapid convergence of our method.

\begin{figure}[t]
    \centering
    \includegraphics[width=0.9\linewidth]{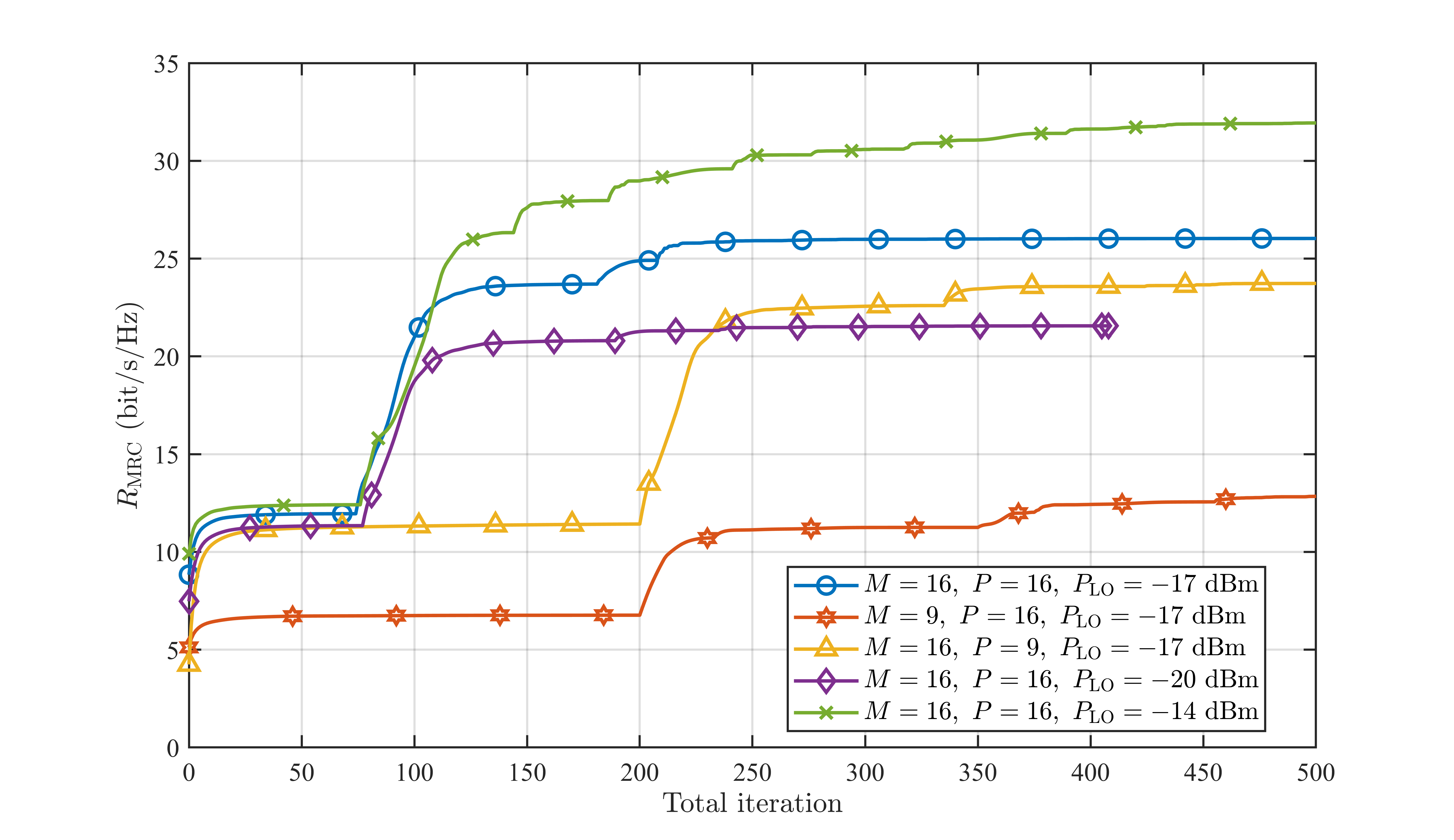}
    \caption{Convergence of our alternating algorithm.}
    \label{convergencr}
     \vspace{-0.5cm}
\end{figure}

We then compare our methods with the following benchmarks.
\begin{enumerate}
    \item \textbf{Benchmark 1}: The optical positions and the LO coefficients are randomly generated within their corresponding
    feasible sets;
    \item \textbf{Benchmark 2}: The genetic algorithm is employed to jointly optimize the LO and optical position;
    \item \textbf{Benchmark 3}: We fix the optical position and optimize the LO by using our gradient ascent algorithm \cite{peng2026celllevelchannelshapingrydberg};
     \item \textbf{Benchmark 4}: The optical position is designed based on the gradient ascent algorithm with a fixed LO;
    \item \textbf{Benchmark 5}: Assuming that there is a far-field LO, the LO and optical position are optimized by using our method;
    \item \textbf{Benchmark 6}: We utilize an array with the same number of conventional fixed-position RF antennas as a baseline.
\end{enumerate}

In Fig. \ref{rateCDF}, we depict the cumulative distribution function (CDF) of sum rate by averaging 100 trials. We assume that users are randomly distributed within an area of 100 \(\rm{km}^2\). As observed, the proposed algorithm approaches the sum-rate performance achieved by the GA while requiring fewer objective evaluations, but the computational complexity of our proposed algorithm is lower than that of the GA, thereby demonstrating its effectiveness. Furthermore, compared to methods that optimize only the LO or optical position, our proposed algorithm reduces interference between users through conversion gain and phase matching, resulting in improved data rates. Compared to a fixed far-field LO, a near-field LO provides a more programmable conversion gain or equivalent channel design, thereby enhancing performance.

\begin{figure}
    \centering
    \includegraphics[width=0.9\linewidth]{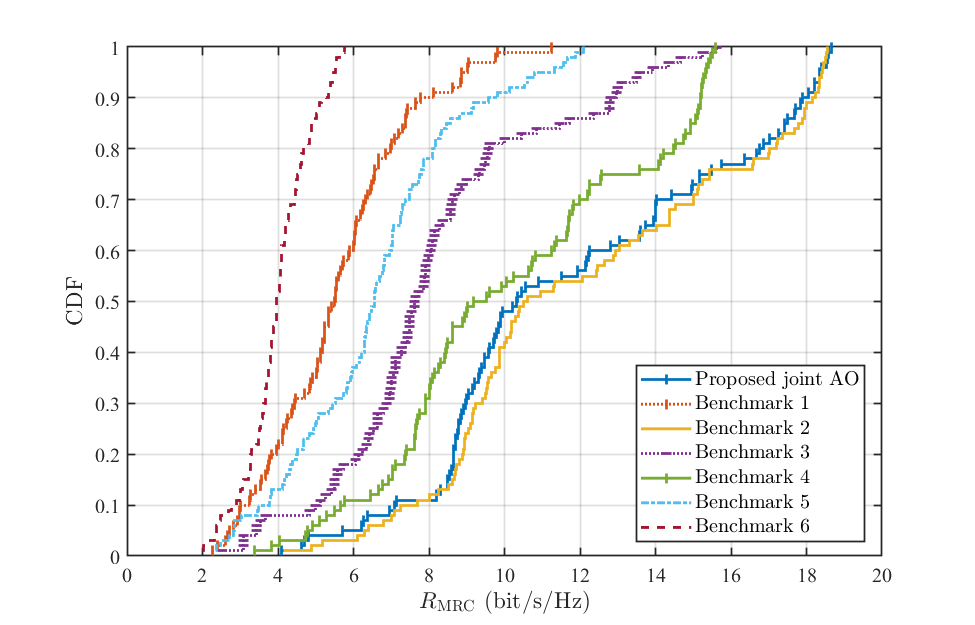}
    \caption{CDF of sum rate based on MRC detection.}
    \label{rateCDF}
     \vspace{-0.5cm}
\end{figure}

Assuming that  \(K\)  users are randomly distributed within 100 \(\rm{km}^2\), Fig. \ref{ratevsk} illustrates the impact of the number of users on the sum rate by averaging over 100 random channel realizations. We find that the sum rate initially increases as the number of users, and then decreases slightly. This occurs because spatially separating users using the available DoFs becomes challenging when the density of users increases, thereby leading to performance degradation.  Importantly, the proposed algorithm outperforms the benchmarks, validating the effectiveness of our algorithm.

\section{Conclusions}

This paper developed an optically movable RAQR architecture by steering the probe–coupling overlap region to dynamically reconfigure the effective RF sensing position without mechanical actuation. A closed-form equivalent baseband model was derived, and its accuracy was validated against numerical solutions of the Lindblad master equation. Based on the derived model, we characterized two complementary channel-shaping mechanisms of movable RAQRs, namely intrinsic beam-pattern shaping through RF-to-optical transduction and per-cell phase control enabled by optical displacement. Numerical results demonstrated that RF-to-optical transduction and optical displacement suppress interference by 14 dB and 87 dB, respectively. To further evaluate the system-level benefits, an MRC-based sum-rate maximization problem was solved via an alternating optimization algorithm with analytical gradients. Simulation results confirmed significant user-decorrelation and sum-rate improvements enabled by optical movability, highlighting its potential as a programmable receiver architecture for future wireless networks.

\begin{figure}
    \centering
    \includegraphics[width=0.9\linewidth]{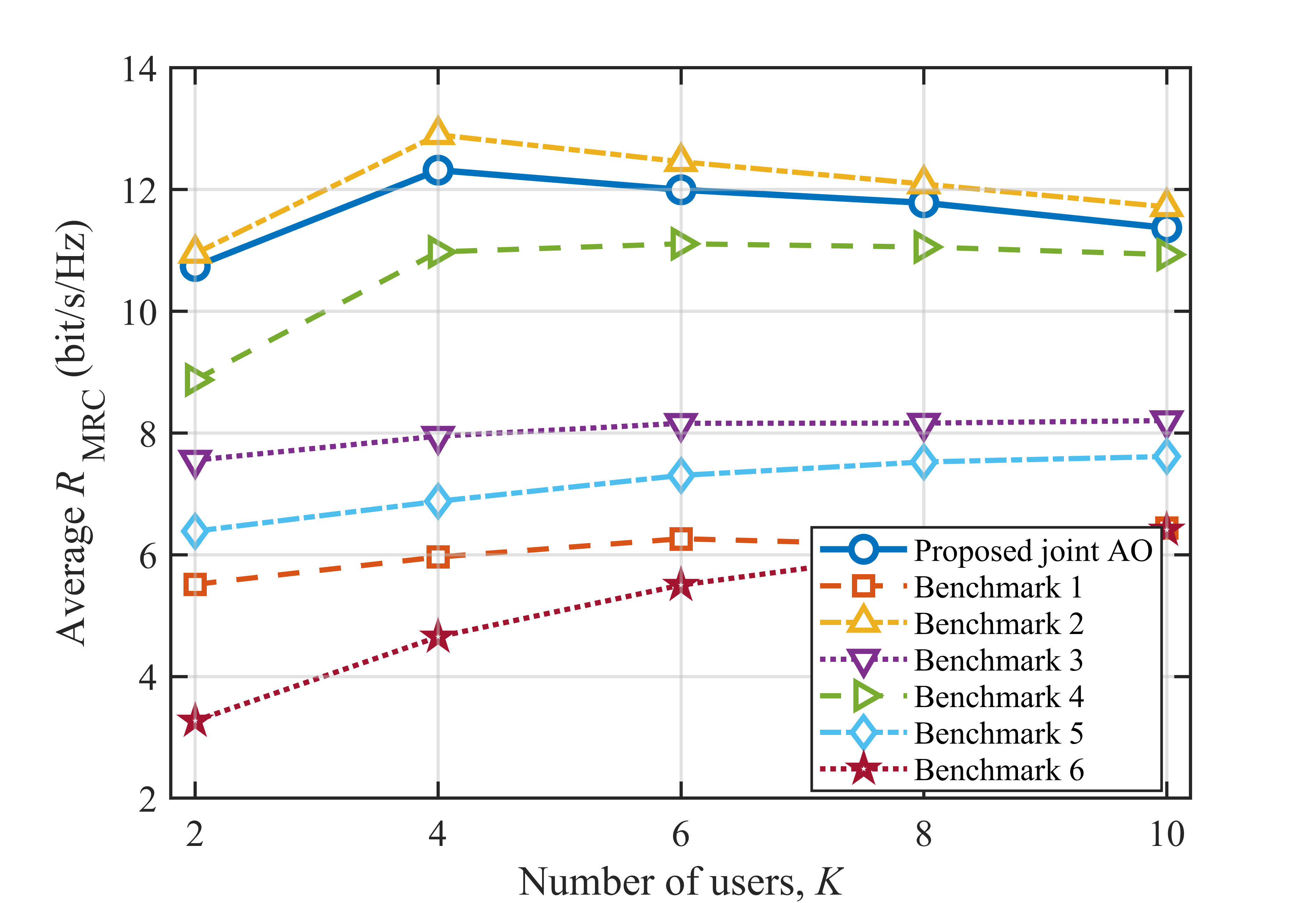}
    \caption{Sum rate under various numbers of users.}
    \label{ratevsk}
   \vspace{-0.5cm}
\end{figure}
\appendices
\section{Proof of Lemma~\ref{lem:mrc_gradient}}
\label{app:mrc_gradient_proof}
Using the chain rule, the partial derivative of the objective function is
\begin{equation}\small
\begin{split}
        &\frac{\partial R_{\rm MRC}(\mathbf Q,\mathbf a)}{\partial\varkappa}
    = \frac{1}{\ln2}
    \sum_{k=1}^{K} \frac{1}{1+{\rm SINR}_k(\mathbf Q,\mathbf a)}\frac{\partial{\rm SINR}_k(\mathbf Q,\mathbf a) }{\partial\varkappa}\\
    = &\frac{1}{\ln2}
    \sum_{k=1}^{K} \frac{\Upsilon_k(\mathbf Q,\mathbf a)}{\Upsilon_k(\mathbf Q,\mathbf a)+\hat\varrho_{kk}^2(\mathbf Q,\mathbf a)}\frac{\partial{\rm SINR}_k(\mathbf Q,\mathbf a) }{\partial\varkappa}\\
   = &  \sum_{k=1}^{K} 
    \Xi_k(\mathbf Q,\mathbf a)
    \Big[
    2\hat\varrho_{kk}(\mathbf Q,\mathbf a) \Upsilon_k(\mathbf Q,\mathbf a)
    \frac{\partial\hat\varrho_{kk}(\mathbf Q,\mathbf a)}{\partial\varkappa}\\
    -&
    \hat\varrho_{kk}^2(\mathbf Q,\mathbf a)
    \frac{\partial \Upsilon_k(\mathbf Q,\mathbf a)}{\partial \varkappa}
    \Big],
\end{split}
    \label{eq:mrc_rate_chain}
\end{equation}
where \( \Xi_k(\mathbf Q,\mathbf a) \) is given by \eqref{eq:mrc_scalar_defs}.
\( \frac{\partial \Upsilon_k(\mathbf Q,\mathbf a)}{\partial\varkappa}\) is given by 
\begin{equation}\small
\begin{split}
        \frac{\partial \Upsilon_k(\mathbf Q,\mathbf a)}{\partial\varkappa}
    &=
    2\sum_{j\neq k}
    \Re\Big\{
    \hat\varrho_{kj}^{*}(\mathbf Q,\mathbf a)
    \frac{\partial\hat\varrho_{kj}(\mathbf Q,\mathbf a)}{\partial\varkappa}
    \Big\}\\
    &+
    \sigma^2
    \frac{\partial\hat\varrho_{kk}(\mathbf Q,\mathbf a)}{\partial\varkappa}.
\end{split}
    \label{eq:mrc_dk_derivative}
\end{equation}
Substituting \eqref{eq:mrc_dk_derivative} into \eqref{eq:mrc_rate_chain}, we have
\begin{equation}\small
\begin{split}
&\frac{\partial R_{\rm MRC}(\mathbf Q,\mathbf a)}{\partial\varkappa}
=
\sum_{k=1}^{K}
\Xi_k(\mathbf Q,\mathbf a)
\hat\varrho_{kk}(\mathbf Q,\mathbf a)
S_k(\mathbf Q,\mathbf a)
\frac{\partial\hat\varrho_{kk}(\mathbf Q,\mathbf a)}{\partial\varkappa}\\
&-
2\sum_{k=1}^{K}
\Xi_k(\mathbf Q,\mathbf a)
\hat\varrho_{kk}^2(\mathbf Q,\mathbf a)
\sum_{j\neq k}
\Re\Big\{
\hat\varrho_{kj}^{*}(\mathbf Q,\mathbf a)
\frac{\partial\hat\varrho_{kj}(\mathbf Q,\mathbf a)}{\partial\varkappa}
\Big\},
\end{split}
\label{eq:mrc_two_terms}
\end{equation}
where \(S_k(\mathbf Q,\mathbf a)\) is given by \eqref{eq:mrc_scalar_defs}.
Note that \(S_k(\mathbf Q,\mathbf a)\) is strictly positive as
$\hat\varrho_{kk}(\mathbf Q,\mathbf a)
=\|\mathbf h_k(\mathbf Q,\mathbf a)\|_2^2>0$. 
With \( \hat\varrho_{kj}(\mathbf Q,\mathbf a) ={\mathbf h}_k^{\rm H}(\mathbf Q,\mathbf a){\mathbf h}_j(\mathbf Q,\mathbf a) \), we have
\begin{equation}\small
\begin{split}
        \frac{\partial\hat\varrho_{kj}(\mathbf Q,\mathbf a) }{\partial\varkappa}
    &=
    \frac{\partial{\mathbf h}_k^{\rm H}(\mathbf Q,\mathbf a) }{\partial\varkappa}
    {\mathbf h}_j(\mathbf Q,\mathbf a) 
    +
    {\mathbf h}_k^{\rm H}(\mathbf Q,\mathbf a) 
    \frac{\partial{\mathbf h}_j(\mathbf Q,\mathbf a) }{\partial\varkappa},\\
    \frac{\partial\hat\varrho_{kk}}{\partial\varkappa}
    &=
    2\Re\Big\{
    {\mathbf h}_k^{\rm H}(\mathbf Q,\mathbf a)
    \frac{\partial{\mathbf h}_k(\mathbf Q,\mathbf a)}{\partial\varkappa}
    \Big\}.
\end{split}
    \label{eq:mrc_gram_derivative}
\end{equation}

Then, substituting the
second line of \eqref{eq:mrc_gram_derivative} into \eqref{eq:mrc_two_terms},
the first summation is rewritten as
\begin{equation}\small
\begin{split}
&\sum_{k=1}^{K}
\Xi_k(\mathbf Q,\mathbf a)
\hat\varrho_{kk}(\mathbf Q,\mathbf a)
S_k(\mathbf Q,\mathbf a)
\frac{\partial\hat\varrho_{kk}(\mathbf Q,\mathbf a)}{\partial\varkappa}
\\
&=
\sum_{k=1}^{K}
\Re\Big\{
\Big[
2\Xi_k(\mathbf Q,\mathbf a)
\hat\varrho_{kk}(\mathbf Q,\mathbf a)
S_k(\mathbf Q,\mathbf a)\,
\mathbf h_k(\mathbf Q,\mathbf a)
\Big]^{\rm H}
\\
&\times
\frac{\partial\mathbf h_k(\mathbf Q,\mathbf a)}{\partial\varkappa}
\Big\}.
\end{split}
\label{eq:mrc_energy_collector}
\end{equation}
For the second summation of \eqref{eq:mrc_two_terms}, substituting the
first line of \eqref{eq:mrc_gram_derivative}, we obtain
\begin{equation}\small
\begin{split}
&\Re\Big\{
\hat\varrho_{kj}^{*}(\mathbf Q,\mathbf a)
\frac{\partial\hat\varrho_{kj}(\mathbf Q,\mathbf a)}{\partial\varkappa}
\Big\}
\\
&=
\Re\Big\{
\hat\varrho_{kj}^{*}(\mathbf Q,\mathbf a)
\frac{\partial\mathbf h_k^{\rm H}(\mathbf Q,\mathbf a)}
{\partial\varkappa}
\mathbf h_j(\mathbf Q,\mathbf a)
\Big\}
\\
&+
\Re\Big\{
\hat\varrho_{kj}^{*}(\mathbf Q,\mathbf a)
\mathbf h_k^{\rm H}(\mathbf Q,\mathbf a)
\frac{\partial\mathbf h_j(\mathbf Q,\mathbf a)}
{\partial\varkappa}
\Big\} \\
&\overset{(a)}{=} \Re\Big\{
\Big[
\hat\varrho_{kj}^{*}(\mathbf Q,\mathbf a)\,
\mathbf h_j(\mathbf Q,\mathbf a)
\Big]^{\rm H}
\frac{\partial\mathbf h_k(\mathbf Q,\mathbf a)}{\partial\varkappa}
\Big\} \\
&+ \Re\Big\{
\Big[
\hat\varrho_{kj}(\mathbf Q,\mathbf a)\,
\mathbf h_k(\mathbf Q,\mathbf a)
\Big]^{\rm H}
\frac{\partial\mathbf h_j(\mathbf Q,\mathbf a)}{\partial\varkappa}
\Big\},
\end{split}
\label{eq:mrc_interference_split}
\end{equation}
where \((a)\) is obtained by using $\Re\{x\}=\Re\{x^{*}\}$.


Therefore, every term of
$\partial R_{\rm MRC}(\mathbf Q,\mathbf a)/\partial\varkappa$ takes the
standard form
$\Re\{(\cdot)^{\rm H}
\partial\mathbf h_i(\mathbf Q,\mathbf a)/\partial\varkappa\}$. By
collecting the like terms that multiply the same
$\partial\mathbf h_i(\mathbf Q,\mathbf a)/\partial\varkappa$ and
merging the two off-diagonal streams through the Hermitian property
$\hat\varrho_{ji}(\mathbf Q,\mathbf a)
=\hat\varrho_{ij}^{*}(\mathbf Q,\mathbf a)$, we construct the collector
vector of user $i$ as
\begin{equation}\small
\begin{split}
\mathbf g_i(\mathbf Q,\mathbf a)
&=
2\Xi_i(\mathbf Q,\mathbf a)
\hat\varrho_{ii}(\mathbf Q,\mathbf a)
S_i(\mathbf Q,\mathbf a)\,
\mathbf h_i(\mathbf Q,\mathbf a)
\\
&-
2\sum_{j\neq i}
\Big[
\Xi_i(\mathbf Q,\mathbf a)
\hat\varrho_{ii}^2(\mathbf Q,\mathbf a) +
\Xi_j(\mathbf Q,\mathbf a)\\
&\times 
\hat\varrho_{jj}^2(\mathbf Q,\mathbf a)
\Big]
\hat\varrho_{ij}^{*}(\mathbf Q,\mathbf a)\,
\mathbf h_j(\mathbf Q,\mathbf a).
\end{split}
\label{eq:mrc_collector}
\end{equation}
Therefore, since \eqref{eq:mrc_collector} shows that \(\mathbf g_i(\mathbf Q,\mathbf a)\) is exactly the \(i\)-th column of \(\boldsymbol\Lambda(\mathbf Q,\mathbf a)\), we have
\begin{equation}\small
\begin{split}
    &\frac{\partial R_{\rm MRC}(\mathbf Q,\mathbf a)}{\partial\varkappa}
=\sum_{i=1}^{K}
\Re\{\mathbf g_i^{\rm H}(\mathbf Q,\mathbf a)\frac
{\partial\mathbf h_i(\mathbf Q,\mathbf a)}{\partial\varkappa}\} \\
&=\Re\Big\{
\sum_{m=1}^{M}
\sum_{k=1}^{K}
\big[
\boldsymbol\Lambda(\mathbf Q,\mathbf a)
\big]_{m,k}^{*}
\frac{\partial[\mathbf H(\mathbf Q,\mathbf a)]_{m,k}}
{\partial\varkappa}
\Big\},
\end{split}
\end{equation}
which completes the proof.

\section{Closed-Form Derivatives for the Optical Design}
\label{app:opt_gradient_proof}
Based on $\Omega_{{\rm LO},m}^{c}(\mathbf{q}_{m},\mathbf a)
=\frac{\mu_{34}}{\hbar}
|\mathcal A_{{\rm LO},m}^{c}(\mathbf{q}_{m},\mathbf a)|$,  its derivative with respect to \(\varkappa\), \(\varkappa\in \{y_m,z_m\}\), is
\begin{equation}\small
\begin{split}
    \frac{
\partial\Omega_{{\rm LO},m}^{c}(\mathbf{q}_{m},\mathbf a)
}{
\partial \varkappa
}
&=
\frac{\mu_{34}}
{\hbar
\left|
  \mathcal A_{{\rm LO},m}^{c}(\mathbf{q}_{m},\mathbf a)
\right|}
\Re
\Big\{
\big[
\mathcal A_{{\rm LO},m}^{c}(\mathbf{q}_{m},\mathbf a)
\big]^{*}\\
& \times 
\frac{
\partial\mathcal A_{{\rm LO},m}^{c}(\mathbf{q}_{m},\mathbf a)
}{
\partial \varkappa
}
\Big\},
\end{split}
\label{eq:opt_omega_derivative}
\end{equation}
where \(\frac{\partial\mathcal A_{{\rm LO},m}^{c}}{\partial\varkappa}\) is given by
\begin{equation}\small
\begin{aligned}
&\frac{
\partial\mathcal A_{{\rm LO},m}^{c}(\mathbf{q}_{m},\mathbf a)
}{
\partial \varkappa
}
=
-V_{\rm LO}
\sum_{p=1}^{P}
a_p^{(r)}
e^{-{\rm j}k_{\rm LO}R_{p,m}^{c}(\mathbf{q}_m)}\\
&\times \Bigg\{\frac{\Delta \varkappa_{p,m}^{c}}{[R_{p,m}^{c}(\mathbf{q}_m) ]^3} +
{\rm j}k_{\rm LO}
\frac{\Delta \varkappa_{p,m}^{c}}{[R_{p,m}^{c}(\mathbf{q}_m) ]^2}
\Bigg\}.
\end{aligned}
\label{eq:opt_lo_field_derivative}
\end{equation}
\(\Delta \varkappa_{p,m}^{c}\) is
\begin{equation}\small
\begin{split}
\begin{cases}
 (m_y-1)d_y + y_m -  y_{\rm LO}-(p_y-1)d_{{\rm LO},y},&\varkappa=y_m,\\
z_m -z_{\rm LO} ,&\varkappa=z_m.
\end{cases}
\end{split}
\end{equation}

\(\sinc'(x)\) and \(u_{\varkappa,k}\) are given by
\begin{equation}\small
\sinc'(x)
=
\begin{cases}
\dfrac{x\cos(x)-\sin(x)}{x^2},
&x\neq0,
\\[2mm]
0,
&x=0.
\end{cases}
\label{eq:opt_sinc_derivative}
\end{equation}
and 
\begin{equation}\small
u_{\varkappa,k}
=
\begin{cases}
u_{y,k},&\varkappa=y_m,\\
u_{z,k},&\varkappa=z_m.
\end{cases}
\label{eq:opt_direction_component}
\end{equation}
\(\frac{
\partial\eta_{{\rm LO},m}^{c}(\mathbf{q}_{m},\mathbf a)
}{
\partial \varkappa
}\) is given by
\begin{equation}\small
\begin{aligned}
\frac{
\partial\eta_{{\rm LO},m}^{c}(\mathbf{q}_{m},\mathbf a)
}{
\partial \varkappa
}
&=
{\rm Im}
\Big\{
\frac{
\partial\dot{\mathcal A}_{{\rm LO},m}^{c}(\mathbf{q}_{m},\mathbf a)/\partial \varkappa
}{
\mathcal A_{{\rm LO},m}^{c}(\mathbf{q}_{m},\mathbf a)
}
-\dot{\mathcal A}_{{\rm LO},m}^{c}(\mathbf{q}_{m},\mathbf a)\\
&\times
\frac{
\left(
\partial\mathcal A_{{\rm LO},m}^{c}(\mathbf{q}_{m},\mathbf a)/\partial \varkappa
\right)
}{
\big[
\mathcal A_{{\rm LO},m}^{c}(\mathbf{q}_{m},\mathbf a)
\big]^2
}
\Big\},
\end{aligned}
\label{eq:opt_eta_derivative}
\end{equation}
where \(\dot{\mathcal A}_{{\rm LO},m}^{c}(\mathbf{q}_{m},\mathbf a)\) and \(\frac{
\partial\dot{\mathcal A}_{{\rm LO},m}^{c}(\mathbf{q}_{m},\mathbf a)
}{
\partial \varkappa
}\) are given by 
\begin{equation}\small
\begin{aligned}
\dot{\mathcal A}_{{\rm LO},m}^{c}(\mathbf{q}_{m},\mathbf a)
&=
-V_{\rm LO}
\sum_{p=1}^{P}
a_p
e^{-{\rm j}k_{\rm LO}R_{p,m}^{c}(\mathbf{q}_m)}\\
&\times 
\left[
\frac{\Delta x_{p,m}}{[R_{p,m}^{c}(\mathbf{q}_m) ]^3}
+
{\rm j}k_{\rm LO}
\frac{\Delta x_{p,m}}{[R_{p,m}^{c}(\mathbf{q}_m) ]^2}
\right],
\end{aligned}
\label{eq:LO_longitudinal_field_recalled}
\end{equation}
and
\begin{equation}\small
\begin{aligned}
&\frac{
\partial\dot{\mathcal A}_{{\rm LO},m}^{c}(\mathbf{q}_{m},\mathbf a)
}{
\partial \varkappa
}
=
V_{\rm LO}
\sum_{p=1}^{P}
a_p^{(r)}
\Delta x_{p,m}
\Delta \varkappa_{p,m}^{c}
e^{-{\rm j}k_{\rm LO}R_{p,m}^{c}(\mathbf{q}_m)}
\\
&\times
\Bigg\{
\frac{3}{[R_{p,m}^{c}(\mathbf{q}_m) ]^5}
+
\frac{3{\rm j}k_{\rm LO}}{[R_{p,m}^{c}(\mathbf{q}_m) ]^4}-
\frac{k_{\rm LO}^2}{[R_{p,m}^{c}(\mathbf{q}_m) ]^3}
\Bigg\}.
\end{aligned}
\label{eq:opt_lo_derivative_derivative}
\end{equation}

\section{Closed-Form Derivatives for the LO Design}
\label{app:LO_gradient_derivation}

This appendix provides the closed-form expressions of
$g_m^{c\prime}(\mathbf{q}_{m},\mathbf a)$,
\(\frac{
\partial\Omega_{{\rm LO},m}^{c}(\mathbf{q}^{(r)}_{m},\mathbf{a})
}{
\partial \varkappa
}\), and
$\frac{\partial\eta_{{\rm LO},m}^{c}(\mathbf{q}_{m},\mathbf a)}{\partial\varkappa}$, where
$\varkappa\in\{\beta_p,\varphi_p\}$. All quantities are evaluated at
the fixed optical position $\mathbf q_m=\mathbf q_m^{(r)}$.

$g_m^{c\prime}(\mathbf{q}_{m},\mathbf a)$ is given by

\begin{equation}\small
\begin{aligned}
&g_m^{c\prime}(\mathbf{q}_{m},\mathbf a)
=
-\frac{2\Omega_p}
{
\left[
C_1(\Omega_m^{c})^4
+
C_2(\Omega_m^{c})^2
+
C_3
\right]^3
}\\
&\times
\Big\{
\Big[
5(B_1C_2-B_2C_1)(\Omega_m^{c})^4+
6(B_1C_3-B_3C_1)(\Omega_m^{c})^2
\\
&+
(B_2C_3-B_3C_2)
\Big]
\left[
C_1(\Omega_m^{c})^4
+
C_2(\Omega_m^{c})^2
+
C_3
\right]
\\
&-
4(\Omega_m^{c})^2
\left[
2C_1(\Omega_m^{c})^2+C_2
\right]
\Big[
(B_1C_2-B_2C_1)(\Omega_m^{c})^4
\\
&
+
2(B_1C_3-B_3C_1)(\Omega_m^{c})^2
+
(B_2C_3-B_3C_2)
\Big]
\Big\},
\end{aligned}
\label{eq:app_LO_g_prime}
\end{equation}
where \(\Omega_m^{c}\triangleq \Omega_{{\rm LO},m}^{c}(\mathbf{q}^{(r)}_{m},\mathbf a)\).

\(\frac{
\partial\Omega_m^{c}
\left(
\mathbf q_m^{(r)},\mathbf a
\right)
}{
\partial\beta_p
}\) is given by \eqref{omega_p}, at the bottom of this page.
\begin{figure*}[hb]
    \hrulefill
    \begin{equation}\small
    \label{omega_p}
    \begin{aligned}
    \frac{
\partial\Omega_m^{c}
\left(
\mathbf q_m^{(r)},\mathbf a
\right)
}{\partial\varkappa} =\begin{cases}
\displaystyle
\frac{
\mu_{34}V_{\rm LO}
}{
\hbar
R_{p,m}^{c}(\mathbf q_m^{(r)})
\left|
\mathcal A_{{\rm LO},m}^{c}
\left(
\mathbf q_m^{(r)},\mathbf a
\right)
\right|
}
{\rm Re}
\Bigg\{
\left[
\mathcal A_{{\rm LO},m}^{c}
\left(
\mathbf q_m^{(r)},\mathbf a
\right)
\right]^*
e^{{\rm j}
\left[
\varphi_p
-
k_{\rm LO}R_{p,m}^{c}(\mathbf q_m^{(r)})
\right]}
\Bigg\},
&\varkappa=\beta_p,
\\[7mm]
\displaystyle
-
\frac{
\mu_{34}V_{\rm LO}\beta_p
}{
\hbar
R_{p,m}^{c}(\mathbf q_m^{(r)})
\left|
\mathcal A_{{\rm LO},m}^{c}
\left(
\mathbf q_m^{(r)},\mathbf a
\right)
\right|
}
{\rm Im}
\Bigg\{
\left[
\mathcal A_{{\rm LO},m}^{c}
\left(
\mathbf q_m^{(r)},\mathbf a
\right)
\right]^*
e^{{\rm j}
\left[
\varphi_p
-
k_{\rm LO}R_{p,m}^{c}(\mathbf q_m^{(r)})
\right]}
\Bigg\},
&\varkappa=\varphi_p.
\end{cases}
\end{aligned}
\end{equation}
\hrulefill
\begin{equation}\small
\begin{aligned}
\frac{
\partial
\dot{\mathcal A}_{{\rm LO},m}^{c}
\left(
\mathbf q_m^{(r)},\mathbf a
\right)
}{
\partial\varkappa
}
&=
\begin{cases}
\displaystyle
V_{\rm LO}
e^{{\rm j}
\left[
\varphi_p
-
k_{\rm LO}R_{p,m}^{c}(\mathbf q_m^{(r)})
\right]}
\left[
-\frac{\Delta x_{p,m}}
{\left[R_{p,m}^{c}(\mathbf q_m^{(r)})\right]^3}
-
{\rm j}k_{\rm LO}
\frac{\Delta x_{p,m}}
{\left[R_{p,m}^{c}(\mathbf q_m^{(r)})\right]^2}
\right],
&\varkappa=\beta_p,
\\[7mm]
\displaystyle
{\rm j}V_{\rm LO}\beta_p
e^{{\rm j}
\left[
\varphi_p
-
k_{\rm LO}R_{p,m}^{c}(\mathbf q_m^{(r)})
\right]}
\left[
-\frac{\Delta x_{p,m}}
{\left[R_{p,m}^{c}(\mathbf q_m^{(r)})\right]^3}
-
{\rm j}k_{\rm LO}
\frac{\Delta x_{p,m}}
{\left[R_{p,m}^{c}(\mathbf q_m^{(r)})\right]^2}
\right],
&\varkappa=\varphi_p.
\end{cases}
\end{aligned}
\label{eq:app_LO_dot_A_derivative}
\end{equation}
\end{figure*}

\(\frac{
\partial\eta_{{\rm LO},m}^{c}
\left(
\mathbf q_m^{(r)},\mathbf a
\right)
}{
\partial\varkappa
}\) is
\begin{equation}\small
\begin{aligned}
&\frac{
\partial\eta_{{\rm LO},m}^{c}
\left(
\mathbf q_m^{(r)},\mathbf a
\right)
}{
\partial\varkappa
}
=
{\rm Im}
\Bigg\{
\frac{
\dfrac{
\partial
\dot{\mathcal A}_{{\rm LO},m}^{c}
\left(
\mathbf q_m^{(r)},\mathbf a
\right)
}{
\partial\varkappa
}
}{
\mathcal A_{{\rm LO},m}^{c}
\left(
\mathbf q_m^{(r)},\mathbf a
\right)
}
\\
&\qquad-
\frac{
\dot{\mathcal A}_{{\rm LO},m}^{c}
\left(
\mathbf q_m^{(r)},\mathbf a
\right)
\dfrac{
\partial
\mathcal A_{{\rm LO},m}^{c}
\left(
\mathbf q_m^{(r)},\mathbf a
\right)
}{
\partial\varkappa
}
}{
\left[
\mathcal A_{{\rm LO},m}^{c}
\left(
\mathbf q_m^{(r)},\mathbf a
\right)
\right]^2
}
\Bigg\},
\end{aligned}
\label{eq:app_LO_eta_derivative}
\end{equation}
where \(\dfrac{
\partial
\mathcal A_{{\rm LO},m}^{c}
\left(
\mathbf q_m^{(r)},\mathbf a
\right)
}{
\partial\varkappa
}\) is given by
\begin{equation}\small
\frac{
\partial\mathcal A_{{\rm LO},m}^{c}(\mathbf q_m^{(r)},\mathbf a)
}{
\partial \varkappa
}
= \begin{cases}
    \displaystyle
    V_{\rm LO}
\frac{
e^{{\rm j}[\varphi_p-k_{\rm LO}R_{p,m}^{c}(\mathbf{q}_m^{(r)})] }
}{
R_{p,m}^{c}(\mathbf{q}_m^{(r)})
},&\varkappa=\beta_p,\\[7mm]
\displaystyle
{\rm j}V_{\rm LO}
\frac{a_p
e^{-{\rm j}k_{\rm LO}R_{p,m}^{c}(\mathbf{q}_m^{(r)})}
}{
R_{p,m}^{c}(\mathbf{q}_m^{(r)})
},&\varkappa=\varphi_p.
\end{cases}
\label{eq:LO_field_real_derivative}
\end{equation}
\( 
\dfrac{
\partial
\dot{\mathcal A}_{{\rm LO},m}^{c}
\left(
\mathbf q_m^{(r)},\mathbf a
\right)
}{
\partial\varkappa}\) is given by  \eqref{eq:app_LO_dot_A_derivative}, at the bottom of this page.

\bibliographystyle{IEEEtran}
\bibliography{myref}

\end{document}